\documentclass[12pt,reqno,titlepage]{amsart}
\usepackage{amssymb,mathrsfs,color,bbold,comment,enumerate,graphicx}
\usepackage[all]{xy}
\usepackage{color}
\usepackage{amsaddr}

\makeatletter
\renewcommand{\email}[2][]{
\ifx\emails\@empty\relax\else{\g@addto@macro\emails{,\space}}\fi
\@ifnotempty{#1}{\g@addto@macro\emails{\textrm{(#1)}\space}}
\g@addto@macro\emails{#2}
}
\makeatother

\theoremstyle{plain}
\newtheorem{theorem}{Theorem}
\newtheorem{proposition}[theorem]{Proposition}
\newtheorem{corollary}[theorem]{Corollary}
\newtheorem{lemma}[theorem]{Lemma}

\theoremstyle{definition}
\newtheorem{remark}[theorem]{Remark}
\newtheorem{example}[theorem]{Example}
\newtheorem{definition}[theorem]{Definition}

\newcommand{\abs}[1]{\lvert#1\rvert}
\newcommand{\norm}[1]{\lVert#1\rVert}

\DeclareMathOperator{\rec}{\rm rec}
\DeclareMathOperator{\lin}{\rm lin}
\DeclareMathOperator{\band}{\rm band}
\DeclareMathOperator{\barr}{\rm bar}
\DeclareMathOperator{\cl}{\rm cl}

\DeclareMathOperator{\VaR}{\rm VaR}
\DeclareMathOperator{\ES}{\rm ES}

\newcommand{\E}{{\mathbb E}}
\newcommand{\N}{{\mathbb N}}
\newcommand{\probp}{{\mathbb P}}
\newcommand{\probq}{{\mathbb Q}}
\newcommand{\R}{{\mathbb R}}

\newcommand{\cA}{\mathcal A}
\newcommand{\cB}{\mathcal B}
\newcommand{\cC}{\mathcal C}
\newcommand{\cD}{\mathcal D}

\newcommand{\cF}{\mathcal F}
\newcommand{\cI}{\mathcal I}
\newcommand{\cL}{\mathcal L}

\newcommand{\cP}{\mathcal P}
\newcommand{\cS}{\mathcal S}
\newcommand{\cX}{\mathcal X}
\newcommand{\cZ}{\mathcal Z}

\def\one{\mathbb 1}

\title{Surplus-Invariant Risk Measures}

\author{Niushan Gao}
\address{Department of Mathematics, Ryerson University, Toronto, Canada}
\email{niushan@ryerson.ca}

\author{Cosimo Munari}
\address{Center for Finance and Insurance and Swiss Finance Institute\\
University of Zurich, Switzerland}
\email{cosimo.munari@bf.uzh.ch}

\keywords{surplus invariance, risk measures, acceptance sets, dual representations, extensions, canonical model spaces, robust model spaces}
\begin{document}

\parindent 0em \noindent

\begin{abstract}
This paper presents a systematic study of the notion of surplus invariance, which plays a natural and important role in the theory of risk measures and capital requirements. So far, this notion has been investigated in the setting of some special spaces of random variables. In this paper we develop a theory of surplus invariance in its natural framework, namely that of vector lattices. Besides providing a unifying perspective on the existing literature, we establish a variety of new results including dual representations and extensions of surplus-invariant risk measures and structural results for surplus-invariant acceptance sets. We illustrate the power of the lattice approach by specifying our results to model spaces with a dominating probability, including Orlicz spaces, as well as to robust model spaces without a dominating probability, where the standard topological techniques and exhaustion arguments cannot be applied.
%In particular, when specialized to topological vector lattices, our results hold without assuming order continuity, a property that is systematically assumed, implicitly or explicitly, in the existing literature.
%The greater level of generality makes our results applicable beyond the standard model spaces.
\end{abstract}

\maketitle

%%%%%%%%%%%%%%%%%%%%%%%%%%%%%%%%%%%%%

\section{Introduction}

%This paper presents a systematic study of the notion of {\em surplus invariance}, which has been recently investigated by several authors, sometimes under different names, in the context of risk measures and capital requirements. So far, the notion of surplus invariance has been studied in the setting of special spaces of random variables. A distinguishing aspect of our approach is that we highlight the lattice foundations of the concept of surplus invariance and develop a theory of surplus invariance by relying solely on lattice properties. This ensures applications to a wide spectrum of standard and less standard model spaces. Besides harmonizing the existing literature, we provide a variety of new results for surplus-invariant acceptance sets and risk measures. In particular, when specialized to topological vector lattices, our results hold without assuming order continuity, a property that is systematically used, implicitly or explicitly, in the existing literature on surplus invariance.
%In particular, our analysis of order closedness allows us to derive powerful dual characterizations of ``surplus-invariant'' risk functionals, which constitute an enriching contribution to the theory of surplus-invariant risk measures.
\subsection{Surplus-invariant acceptance sets}

To introduce the notion of surplus invariance in a financial context, take the perspective of an internal or external financial supervisor who has to design a test to assess whether the risk of a future financial position is acceptable or not. If we represent financial positions by means of the elements of a vector lattice $\cX$, the test can be naturally identified with a subset $\cA\subset\cX$: A position will be deemed acceptable if, and only if, it belongs to $\cA$. In the theory of risk measures one often refers to $\cA$ as an {\em acceptance set}; see Artzner et al.~(1999). The main objective of the theory is to formalize and study desirable properties of acceptance sets. While the assessment of desirability will clearly depend on the specific financial application, there is at least one universally desirable property, namely {\em monotonicity}. If we interpret, as we will always do in the sequel, a positive element of $\cX$ as a payoff
(profit or surplus) and a negative element of $\cX$ as a payout (loss or default), then we say that $\cA$ is monotone whenever
\[
X\in\cA, \ Y\geq X \ \implies \ Y\in\cA.
\]
This property formalizes the basic idea that every position whose profit-loss pattern is better than that of an acceptable position should also be acceptable. Another property that is often seen as desirable is {\em convexity}, which reflects the diversification principle according to which a ``portfolio'' of acceptable positions should remain acceptable; see F\"{o}llmer and Schied (2002) and Frittelli and Rosazza Gianin (2002).
%{\em conicity}. The latter property makes acceptability independent of the position's size.
%The combination of convexity and conicity is usually referred to as {\em coherence}.
%If $\cX$ is a space of random variables, then distributional properties such as {\em law invariance} or {\em stochastic-order monotonicity} may also come into play, see F\"{o}llmer and Schied (2011).

\smallskip

The property of {\em surplus invariance} stipulates that acceptability does not depend upon the (size of the) positive part of a financial position. Formally, we say that $\cA$ is surplus invariant whenever
\[
X\in\cA, \ \min(Y,0)=\min(X,0) \ \implies \ Y\in\cA.
\]
%Under monotonicity, surplus invariance becomes equivalent to
%\[
%X\in\cA, \ \min(Y,0)\geq\min(X,0) \ \implies \ Y\in\cA,
%\]
%which can be interpreted as a form of monotonicity at the level of the negative parts of financial positions. Recall that we interpret negative parts as payouts (in absolute terms). At this point, one should ask:
Whether surplus invariance is desirable or not will depend, once again, on the particular financial application. If the threshold to acceptability is defined by a tradeoff between risk and return, so that the benefits engendered by the profits are allowed to compensate for the costs imposed by the losses, then surplus invariance is clearly off the target. Acceptance sets based on risk-return tradeoffs are often encountered in pricing theory; see Cochrane and Sa\'{a}-Requejo (2000), Bernardo and Ledoit (2000), Carr et al.~(2001), and more recently Madan and Cherny (2010). However, if acceptability has to capture the downside risk of a financial position only, then surplus invariance is a natural and desirable property. As illustrated by the following examples, the exclusive focus on downside risk is meaningful if we are in a situation where the benefits stemming from the profits and the burdens caused by the losses are enjoyed, respectively suffered, by different stakeholders. Here, we denote by $\cX$ a suitable space of random variables on a given probability triple $(\Omega,\cF,\probp)$.

\smallskip

{\bf Capital requirements}. When designing capital adequacy tests for financial institutions with limited liability, regulators must take into account the following fundamental asymmetry. On the one side, if the company is solvent, then all the liabilities are paid off and the company will generally enjoy a surplus, which accrues to the shareholders (and can be distributed as dividends and/or reivested in future business). In this case, taking the perspective of a one-period model as usually done in capital adequacy, the size of the surplus should not matter to the liability holders while it is a crucial figure for the shareholders. On the other side, if the company is insolvent, then the shareholders may exercise their limited liability option and are under no obligation to finance the deficit. In this case, the shareholders will be in principle indifferent to the size of default, which is entirely suffered by the liability holders. Hence, surplus invariance seems to be aligned with the fundamental objective of solvency regulation, namely the protection of liability holders. A variety of results for surplus-invariant acceptance sets in a capital adequacy context can be found in Staum (2013), Koch-Medina et al.~(2015), Liu and Wang (2016), He and Peng (2017), Koch-Medina et al.~(2017), and Bignozzi et al.~(2018). In particular, acceptance sets based on Value at Risk
\[
\cA = \{X\in\cX \,; \ \VaR_\alpha(X)\leq0\}, \ \ \ \VaR_\alpha(X)=\inf\{m\in\R
\,; \ \probp(X+m<0)\leq\alpha\},
\]
are surplus invariant for every $\alpha\in(0,1)$ while those based on Expected Shortfall
\[
\cA = \{X\in\cX \,; \ \ES_\alpha(X)\leq0\}, \ \ \
\ES_\alpha(X)=\frac{1}{\alpha}\int_0^\alpha\VaR_\beta(X)d\beta,
\]
fail to be surplus invariant for every level $\alpha\in(0,1)$. We refer to Koch-Medina and Munari (2016) for a thorough discussion about the financial implications of the failure of surplus invariance for $\ES$-based acceptance sets.
%Acceptability based on $\VaR$ is equivalent to requiring solvency in at least $100(1-\alpha)\%$ of cases while acceptability based on $\ES$ is roughly equivalent to requiring solvency on average in the worst $100\alpha\%$ of cases. 

\smallskip

{\bf Margin requirements}. Consider a central exchange or a central counterparty (CCP) where transactions between multiple counterparties are netted. Each counterparty is required to hold a margin deposit in order to ensure sufficient liquidity to absorb unexpected losses. As argued in Cont et al.~(2013), from a CCP's perspective only downside risk matters and, therefore, the underlying capital adequacy test should be represented by a surplus-invariant acceptance set. The SPAN (Standard Portfolio ANalysis) methodology for the computation of margin requirements, which was developed by the Chicago Mercantile Exchange and is implemented by many other exchanges, is based on surplus-invariant acceptance sets of the form
\[
\cA = \{X\in\cX \,; \ \probp(\{X<0\}\cap E)=0\},
\]
where $E\in\cF$ is a suitable set of test scenarios. As proved in Koch-Medina et al.~(2017), this is essentially the only surplus-invariant acceptance set that is simultaneously convex and conic, or equivalently coherent in the terminology of Artzner et al.~(1999).
%in the family of coherent acceptance sets.
%In this case, acceptability is equivalent to requiring solvency in every scenario belonging to $E$. 

\smallskip

{\bf Hedging}. Consider a risk-averse agent who wants to hedge some future monetary exposure. In this context, as discussed in F\"{o}llmer and Schied (2002), the agent may want to define a notion of acceptability based on the shortfall of the hedged position. This naturally leads to focusing on surplus-invariant acceptance sets of the form
\[
\cA = \{X\in\cX \,; \ \E[\ell(-\min(X,0))]\leq\alpha\},
\]
where $\ell:\R\to\R$ is a suitable loss function and $\alpha\in\R$. In this case,
acceptability boils down to requiring that the expected ``disutility'' attached to the shortfall of the hedged position is
sufficiently small.

%\medskip
\subsection{Surplus-invariant risk functionals}

So far we have focused on surplus invariance from the perspective of acceptance sets. Once we turn our attention to risk functionals $\rho:\cX\to\overline{\R}$, there are two natural ways to introduce a notion of surplus invariance. On the one side, we say that $\rho$ is {\em surplus invariant} if for every $X\in\cX$
\[
\rho(X)=\rho(\min(X,0))
\]
In this case, the outcome of $\rho$ is assumed to depend exclusively on the downside of a position. This is a sensible definition if $\rho$ is used to rank positions according to their downside risk. The {\em loss-based risk measures} in Cont et al.~(2013) and the {\em shortfall risk measures} in Staum (2013) are examples of measures of downside risk.

\smallskip

On the other side, we say that $\rho$ is {\em surplus invariant subject to positivity} whenever
\[
\rho(X)>0 \ \implies \ \rho(X)=\rho(\min(X,0)).
\]
This weaker notion of surplus invariance, which was introduced in a slightly different guise in Koch-Medina et al.~(2015), is suitable when $\rho$ is employed as a rule to determine capital or margin requirements. In this case, the quantity $\rho(X)$ is interpreted as a capital requirement. The condition $\rho(X)>0$ means that $X$ is not adequately capitalized and requires the injection of risk capital. The above property tells us that, in this case, the amount of risk capital will depend on the downside of the position $X$
only. On the other side, the condition $\rho(X)\leq0$ means that $X$ is already sufficiently capitalized and, hence, $\rho(X)$ can be interpreted as a capital amount that can be returned to shareholders. It is intuitively clear that, in this case, the quantity $\rho(X)$ may (and typically will) depend also on the size of the surplus.
%This interpretation becomes fully convincing for the capital requirement rules introduced in Artzner et al.~(1999), for which $\rho(X)$ represents the ``distance'', measured in terms of a given reference asset, to a pre-specified acceptance set.

\smallskip

It is interesting to mention that surplus invariance was already listed among the defining properties of a coherent risk measure (!) in a Technical Report of Cornell University by Artzner et al.~(1996). The focus of that preprint was on solvency metrics and surplus invariance was argued to be a natural property in that context. In Artzner et al.~(1999) the axiom of surplus invariance was eventually replaced by {\em translation invariance}, according to which for all $X\in\cX$ and $m\in\R$
\[
\rho(X+mS) = \rho(X)-m,
\]
where $S\in\cX$ is a given payoff. Translation invariance allows to express $\rho$ as
\[
\rho(X) = \inf\{m\in\R \,; \ X+mS\in\{\rho\leq0\}\}
\]
for every $X\in\cX$ and thus yields a convenient operational interpretation: $\rho(X)$ can be interpreted as the minimal amount of capital that has to be raised and invested in the asset with payoff $S$ in order to ensure ``acceptability''. The reason to dismiss surplus invariance lies in its fundamental incompatibility with translation invariance. There is, however, no incompatibility between translation invariance and surplus invariance subject to positivity. This shows how to reconcile the original motivation for surplus invariance in a capital adequacy setting with the axioms of coherence.

\subsection{Structure of the paper and main results}
%We provide a comprehensive study of surplus invariance in the framework of vector lattices. This allows us to subsume in a unifying perspective the existing results on surplus invariance, for which a more elementary proof can be sometimes obtained. At the same time, being established in full generality, our new results apply to a vast spectrum of model spaces, which range from the standard spaces of random variables, equipped with a dominating probability, to model spaces with multiple priors or no prior, which constitute the natural working environment in the growing field of ``robust finance''. In particular, when specialized to topological vector lattices, our results hold without assuming order continuity, a property that is systematically used, implicitly or explicitly, in the existing literature on surplus invariance. This implies, e.g., that our results are valid, under a dominating probability, in the context of any Orlicz space.
%In this paper we focus on applications to the standard model spaces.
%The paper is structured as follows.

In Section 2 we introduce the notion of surplus invariance for subsets of a general vector lattice. The main result is the characterization of order closedness recorded in Theorem~\ref{closed-top}, which constitutes the key to deriving a number of powerful dual representations for surplus-invariant risk measures. In addition, we establish a general decomposition result for surplus-invariant acceptance sets in Theorem~\ref{decom} and use it to derive the existence of strictly-positive supporting functionals in Theorem~\ref{theo: strictly positive supporting functional}. In Section 3 we take up the study of surplus-invariant functionals. After highlighting the link with surplus-invariant sets, we apply our results to establish dual characterizations of the Fatou property for quasiconvex surplus-invariant functionals in Theorems~\ref{s-i functionals} and~\ref{theo: dual representation S additive}. Moreover, we discuss extension results for surplus-invariant functionals in Theorems~\ref{theo: general extension} and~\ref{theo: extension si subject pos}.

\smallskip

We devote Section 4 to applications. We focus on model spaces with and without a dominating probability. In the latter case we adopt the setting of Maggis et al.~(2018). The generality of our lattice approach allows us to deal with the nondominated setting essentially with no extra effort.

\smallskip

{\bf Fatou property and duality}. It follows from the far-reaching result by Delbaen (2002) that every convex functional that satisfies the Fatou property on $L^\infty(\probp)$ is automatically $\sigma(L^\infty(\probp),L^1(\probp))$ lower semicontinuous and, hence, admits a ``nice'' dual representation where the dual elements belong to $L^1(\probp)$ and can be therefore identified with countably-additive signed measures. As recently shown in Gao et al.~(2016), the Fatou property will not generally ensure a ``nice'' dual representation on a general Orlicz space $L^\Phi(\probp)$. A robust version of Delbaen's result has been recently obtained in Maggis et al.~(2018) in a bounded setting.
%, i.e.\ a dual representation where the dual domain is a subspace of $L^1(\probp)$, if we move from $L^\infty(\probp)$ to a general Orlicz space $L^\Phi(\probp)$. We also refer to Delbaen and Owari (2016) and, for a comprehensive picture, to Gao et al.~(2016). 

\smallskip

In the context of surplus-invariant functionals, we provide a variety of dual characterizations of the Fatou property that refine and extend the results by Delbaen (2002) and Maggis et al.~(2018) beyond the bounded setting. In Theorem \ref{theo: fatou in L infty} we start by showing that surplus invariance on $L^\infty(\probp)$ allows to restrict the set of dual elements from $L^1(\probp)$ to the smaller space $L^\infty(\probp)$. This result leads to an improved version of the dual representation in Cont et al.~(2013). In Theorem \ref{theo: fatou orlicz spaces} we prove that, under the assumption of surplus invariance, the Fatou property on a general Orlicz space $L^\Phi(\probp)$ {\em
always} implies ``nice'' dual representations where the dual elements can be taken to belong to $L^{\Phi^\ast}(\probp)$ or $H^{\Phi^\ast}(\probp)$, where $\Phi^\ast$ denotes the conjugate of $\Phi$, or even $L^\infty(\probp)$. 

\smallskip

Finally, we focus on the robust case and extend the result in Maggis et al.~(2018) in two ways. First, we consider robust model spaces beyond the bounded case. Second, we show that, under surplus invariance (subject to positivity), the Fatou property yields a dual representation where the dual elements can be identified with countably-additive signed measures that have bounded Radon-Nikodym derivatives and that are directly associated with the underlying family of probability measures. 

\smallskip

Both in the dominated and in the nondominated setting, we actually show that the Fatou property
is equivalent to the (much) stronger property of lower semicontinuity with respect to unbounded-order convergence (which corresponds to almost-sure convergence in the dominated case).

\smallskip

{\bf Extension results}. In Theorem \ref{theo: robust extensions} we show that every (quasi)convex monotone functional with the Fatou property defined on $L^\infty(\probp)$ can be uniquely extended, under surplus invariance (subject to positivity), to the entire space $L^0(\probp)$ without losing its properties. In addition, the extension to $L^1(\probp)$ is also unique and can be characterized in dual terms. This shows that the extension problem formulated in Delbaen (2002) can be always tackled under surplus invariance (subject to positivity). The extension result holds even if there exists no dominating probability.

\smallskip

{\bf Decomposition results}. A variety of decomposition results for surplus-invariant acceptance sets have been obtained in Koch-Medina et al.~(2017) in the context of subspaces of $L^0(\probp)$ containing $L^\infty(\probp)$ in a dense way. In Theorem \ref{theo: robust decomposition} we extend the decomposition to every space of random variables in a dominated setting. In fact, using our general results, we are able to extend the decomposition to robust model spaces with no extra effort.

\smallskip

{\bf Bipolar theorems}. As a final application, we exploit the link between surplus-invariant sets and solid subsets of the positive cone to provide a new lattice-based proof of the bipolar theorem by Brannath and Schachermayer (1999) in Theorem~\ref{theo: bipolar} and to extend it to the robust setting in Theorem \ref{theo: robust bipolar}. It is worth remarking that the classical proof is based on an exhaustion argument that cannot be applied in the robust setting due to the failure of the countable sup property.
%This highlights the power of the lattice approach.

\smallskip

A final appendix collects standard terminology and notation related to vector lattices.

%%%%%%%%%%%%%%%%%%%%%%%%%%%%%%%%%%%%%

\section{Surplus-invariant monotone sets}

We start by introducing the notion of surplus invariance and establishing a variety of basic properties of
surplus-invariant monotone sets. In a convex setting, we discuss the link between order and topological closedness, which will be
later exploited to derive a variety of dual representations for risk functionals, and establish a decomposition result for
surplus-invariant monotone sets. Throughout the section we assume that $\cX$ is a fixed vector lattice. We refer to the appendix for
an overview of standard terminology and notation related to vector lattice.

\begin{definition}
A subset $\cA\subset\cX$ is {\em surplus invariant} if for every $X,Y\in\cX$ we have
\[
X\in\cA, \ X^-=Y^- \ \implies \ Y\in\cA.
\]

\smallskip

We say that $\cA$ is {\em monotone} if for every $X,Y\in\cX$ we have
\[
X\in\cA, \ Y\geq X \ \implies \ Y\in\cA.
\]
\end{definition}

\smallskip

The property of surplus invariance was introduced in a mathematical finance context by Staum~(2013) under the name of {\em excess
invariance} and was later studied by Koch-Medina et al.~(2015) and Koch-Medina et al.~(2017). The property of monotonicity is standard
in the theory of acceptance sets and risk measures, see F\"{o}llmer and Schied~(2017).

\smallskip

The following proposition collects a variety of equivalent conditions for a (monotone) set to be surplus invariant. In particular, every surplus-invariant monotone set can be decomposed into the difference between the positive cone and a solid subset of the positive cone. This highlights the intimate link between surplus invariance and solidity.

\begin{proposition}
\label{s-i}
For a set $\cA\subset\cX$ the following statements are equivalent:
\begin{enumerate}
	\item[(a)] $\cA$ is surplus invariant.
	\item[(b)] $\cA=\{X\in\cX \,; \ X^-\in\cD\}$ for some $\cD\subset\cX_+$.
\end{enumerate}
If $\cA$ is monotone, the above are also equivalent to:
\begin{enumerate}
  \item[(c)] $X\in\cA$ implies $-X^-\in\cA$.
  \item[(d)] $X\in\cA$ and $Y^-\leq X^-$ imply $Y\in\cA$.
  \item[(e)] $\cA=\cX_+-\cD$ for some $\cD\subset\cX_+$ that is solid in $\cX_+$.
\end{enumerate}
In both (b) and (e) we have $\cD=-\cA_-$.
\end{proposition}
\begin{proof}
It is clear that {\em (b)} implies {\em (a)}. To prove the converse implication, assume that {\em (a)} holds and note that $X$ and
$-X^-$ have the same negative parts for each $X\in\cX$. Thus, it follows from surplus invariance that $X\in\cA$ if and only if $-X^-\in\cA_-$, or equivalently $X^-\in-\cA_-$. This shows that {\em (b)} holds and $\cD=-\cA_-$.

\smallskip

Now, assume that $\cA$ is monotone. Again, since the negative parts of $X$ and $-X^-$ coincide for every $X\in\cX$, it is clear that {\em
(a)} implies {\em (c)}. Next, assume that {\em (c)} holds and
take $X\in\cA$ and $Y\in\cX$ such that $Y^-\leq X^-$. Since $-X^-\in\cA$ by
assumption and $Y\geq-Y^-\geq-X^-$, we must have $Y\in\cA$ by
monotonicity. This shows that {\em (d)} holds. Clearly, {\em (d)} implies {\em (a)}.

\smallskip

Now, assume that {\em (c)} holds and set $\cD=-\cA_-$. By monotonicity, we clearly have $\cX_+-\cD\subset\cA$. On the other hand, for
every $X\in\cA$, we have $X^-\in\cD$ by assumption and thus it follows from $X=X^+-X^-$ that $\cA\subset\cX_+-\cD$. Therefore,
$\cA=\cX_+-\cD$. Now, take $0\leq Y\leq X\in
\cD$. Since $-Y\geq-X\in\cA_-$, we infer that $-Y\in \cA_-$ by monotonicity and thus $Y\in\cD$. This proves that $\cD$ is solid in
$\cX_+$ and establishes {\em (e)}. Conversely, assume that {\em (e)} holds and take an arbitrary $X\in\cA$. By assumption we find
$Y\in\cX_+$ and $Z\in\cD$ such that $X=Y-Z$. Then, we easily see that
\[
0 \leq X^- = \max(-X,0) = \max(Z-Y,0) \leq Z \in \cD,
\]
which implies $X^-\in\cD\subset-\cA_-$ by solidity of $\cD$. This shows that $-X^-\in\cA$ and proves that {\em (c)} holds. In fact, this also shows that for every $X\in\cA_-$ we have $-X=X^-\in\cD$ and thus $-\cA_-\subset\cD$, showing that {\em (e)} can only hold if $\cD=-\cA_-$.
\end{proof}

\smallskip

As a direct consequence of the preceding result, one can show that the convexity and conicity of a surplus-invariant monotone set can
be characterized at the level of its negative part. Recall that a set $\cA\subset\cX$ is said to be {\em convex} whenever $\lambda
X+(1-\lambda)Y\in\cA$ for every $\lambda\in[0,1]$ and $X,Y\in\cA$. We say that $\cA$ is a {\em cone} if $\lambda X\in\cA$ whenever
$\lambda\geq0$ and $X\in\cA$.

\begin{corollary}
\label{cor: convexity}
For a surplus-invariant monotone set $\cA\subset\cX$ we have:
\begin{enumerate}
  \item[(a)] $\cA$ is convex iff $\cA_-$ is convex.
  \item[(b)] $\cA$ is a cone iff $\cA_-$ is a cone.
\end{enumerate}
\end{corollary}
\begin{proof}
The ``only if'' implications are clear because $\cX_-$ is a convex cone. The ``if'' implications are immediate once we observe that
$\cA=\cX_++\cA_-$ by Proposition~\ref{s-i}.
\end{proof}

\smallskip

In the case that $\cX$ has the projection property, so that every band is automatically a projection band, we find the following
useful characterization of surplus invariance in terms of projection operators.

\begin{proposition}
\label{prop: si and projections}
Assume $\cX$ has the projection property. Then, for a monotone set $\cA\subset\cX$ the following statements are equivalent:
\begin{enumerate}
  \item[(a)] $\cA$ is surplus invariant.
  \item[(b)] $X\in\cA$ implies $P(X)\in\cA$ for every band projection $P$ on $\cX$.
\end{enumerate}
\end{proposition}
\begin{proof}
Assume first that {\em (a)} holds and let $X\in\cA$ and $P$ be a band projection on $\cX$. Then, we have
\[
P(X) = P(X^+)-P(X^-) \geq -P(X^-) \geq -X^-.
\]
Since $-X^-\in\cA$ by Proposition~\ref{s-i}, it follows from monotonicity that $P(X)\in\cA$
so that {\em (b)} holds. Next, assume that {\em (b)} holds. Take any $X\in\cA$ and let $P$
be the band projection for the principal band $\cB_{X^-}$. Since
$X^-\in\cB_{X^-}$ by definition, we clearly have $P(X^-)=X^-$. In addition,
since $X^+\wedge X^-=0$, we infer that $X^+\wedge\abs{Y}=0$ for all
$Y\in\cB_{X^-}$ so that $P(X^+)=0$. As a result, it follows from {\em (b)} that
\[
-X^- = P(X^+)-P(X^-) = P(X)\in \cA.
\]
In view of Proposition~\ref{s-i}, this implies that {\em (a)} holds and concludes the proof.
\end{proof}

\smallskip

%%%%%%%%%%%%%%%%%%%%%%%%%%%%%%%%%%%%%%%%%%%%%%

\subsection{Closedness of surplus-invariant monotone sets}

In this section we investigate the link among some (order and topological)
closedness properties of surplus-invariant monotone sets. First, we show that (order and
topological) closedness properties of a surplus-invariant monotone set are
systematically captured by the closedness of the corresponding negative parts.

\begin{proposition}
\label{prop: closedness}
For a surplus-invariant monotone set $\cA\subset\cX$ we have:
\begin{enumerate}
  \item[(a)] $\cA$ is order closed iff $\cA_-$ is order closed.
  \item[(b)] $\cA$ is uo closed iff $\cA_-$ is uo closed.
\end{enumerate}
Assume $\tau$ is a locally-solid Hausdorff topology on $\cX$. Then, we have:
\begin{enumerate}
  \item[(c)] $\cA$ is $\tau$ closed iff $\cA_-$ is $\tau$ closed.
\end{enumerate}
\end{proposition}
\begin{proof}
Since the limit of an order-convergent net consisting of negative elements of
$\cX$ is still negative, we see that the order closedness of $\cA_-$ follows
from that of $\cA$. The same implication holds if we replace order convergence
by unbounded-order convergence, or by $\tau$ convergence, once we recall that
$\cX_-$ is $\tau$ closed by Theorem~2.21 in Aliprantis and Burkinshaw (2003).
Conversely, assume that $\cA_-$ is order closed and consider a net
$(X_\alpha)\subset\cA$ and $X\in\cX$ satisfying $X_\alpha\xrightarrow{o}X$.
Since $-X^-_\alpha\in\cA$ for all $\alpha$ by surplus invariance and since we
have $-X^-_\alpha\xrightarrow{o}-X^-$, it follows that $-X^-\in\cA$ by order
closedness of $\cA_-$ and, thus, $X\in\cA$ by monotonicity. This proves that
$\cA$ is order closed. The same argument works if we replace order convergence
by unbounded-order convergence, or by $\tau$ convergence in view of the $\tau$
continuity of the lattice operations, see Theorem~2.17 in Aliprantis and
Burkinshaw (2003).
\end{proof}

\smallskip

The second preliminary result highlights some properties of unbounded-order convergence that also carry some independent interest.

\begin{lemma}
\label{con}
Consider a net $(X_\alpha)\subset\cX_+$ and $X\in\cX_+$ such that $X_\alpha\xrightarrow{uo}X$.
\begin{enumerate}
 \item[(a)] If $\cX$ is order complete, then $X=\sup_{\alpha}\inf_{\beta\geq\alpha}X_\beta$.
 \item[(b)] If $(X_\alpha)$ is contained in a set $\cD$ that is solid in $\cX_+$, then there exists an increasing net
     $(Z_\gamma)\subset\cD$ satisfying $X=\sup_\gamma Z_\gamma$.
\end{enumerate}
\end{lemma}
\begin{proof}
To prove {\em (a)}, fix $Y\in\cX_+$ and note that
$\abs{X_\alpha-X}\wedge Y\xrightarrow{o}0$ by definition of unbounded-order
convergence. Since for every $\alpha$ we have
\[
\abs{X_\alpha\wedge Y-X\wedge Y} \leq \abs{X_\alpha-X}\wedge Y,
\]
we infer that $X_\alpha\wedge Y\xrightarrow{o}X\wedge Y$. In turn, this implies
\begin{equation}
\label{eq: lemma con}
X\wedge Y = \sup_{\alpha}\inf_{\beta\geq\alpha}(X_\beta\wedge Y) =
\sup_{\alpha}\big((\inf_{\beta\geq\alpha}X_\beta)\wedge Y\big).
\end{equation}
In particular, for any $Y\in\cX_+$ and for any index $\alpha$ we have
\[
\inf_{\beta\geq\alpha}X_\beta\wedge Y \leq X.
\]
Taking $Y=\inf_{\beta\geq\alpha}X_\beta$, we infer that
$\inf_{\beta\geq\alpha}X_\beta\leq X$ and, hence,
$\sup_{\alpha}\inf_{\beta\geq\alpha}X_\beta$ exists. Therefore, we can apply a
further distribution law in \eqref{eq: lemma con} and obtain
\[
X\wedge Y = \big(\sup_{\alpha}\inf_{\beta\geq\alpha}X_\beta\big)\wedge Y
\]
for any $Y\in\cX_+$. Taking  $Y=X+\sup_{\alpha}\inf_{\beta\geq\alpha}X_\beta$
we  obtain the desired identity.

\smallskip

To prove {\em (b)}, let $\cX^\delta$ be the order completion of $\cX$ and note that
$X_\alpha\xrightarrow{uo}X$ in $\cX^\delta$, see Theorem~3.2 in Gao, Troitsky and Xanthos (2017).
Thus, by point {\em (a)}, we have
\[
X = \sup_\alpha\inf_{\beta\geq\alpha}X_\beta,
\]
where the lattice operations are understood in $\cX^\delta$. For each $\alpha$, define $Z_\alpha=\inf_{\beta\geq\alpha}X_\beta$ and
set
\[
\mathcal{Z}_\alpha=\{X\in \cX \,; \ 0\leq X\leq Z_\alpha\}.
\]
In addition, define the collection
\[
\mathcal{Z} = \bigcup_\alpha\mathcal{Z}_\alpha.
\]
Since for every $\alpha$ we have $Z_\alpha\leq X_\alpha$ and, hence, $\mathcal{Z}_\alpha\subset\cD$ by solidity of $\cD$, we infer
that $\mathcal{Z}\subset\cD$. Moreover, $\mathcal{Z}$ is directed upward in $\cX_+$. To see this, for $i\in\{1,2\}$ consider any
$0\leq X_i\leq Z_{\alpha_i}$ and take $\alpha$ with $\alpha\geq\alpha_i$ for all $i\in\{1,2\}$. Then, we have $X_1\vee X_2\leq
Z_{\alpha_1}\vee
Z_{\alpha_2}\leq Z_\alpha$, and thus $X_1\vee X_2\in\mathcal{Z}$. As a result, $\mathcal{Z}$ can be viewed as an increasing net in
$\cX_+$ indexed over itself.
To conclude the proof, recall that $\cX$ is order dense in $\cX^\delta$ and therefore we have
$\sup\mathcal{Z}_\alpha=Z_\alpha$ where the supremum is taken in $\cX^\delta$.
Thus
\[
\sup\mathcal{Z}=\sup_\alpha\sup \mathcal{Z}_\alpha=\sup_\alpha
Z_\alpha=X,
\]
where the suprema are again taken in $\cX^\delta$. By order denseness of $\cX$
in $\cX^\delta$, it follows that $\sup\mathcal{Z}=X$ in $\cX$ as well, completing the proof.
\end{proof}

\smallskip

The third and last preliminary result follows immediately by combining Lemma~4.2 and
Theorem~4.20 in Aliprantis and Burkinshaw (2003).

\begin{lemma}
\label{lem: amemiya-fremlin}
Assume $\tau$ is an order-continuous locally-solid Hausdorff topology on
$\cX$. Then, for every solid set $\cS\subset\cX$ the following statements are
equivalent:
\begin{enumerate}
	\item[(a)] $\cS$ is order closed.
	\item[(b)] $\cS$ is $\tau$ closed.
\end{enumerate}
\end{lemma}

\smallskip

We are now ready to establish the main result of this section, which provides a twofold characterization of order closedness for surplus-invariant monotone sets. In a first step, we show that order closedness is equivalent to unbounded-order closedness (which is typically much stronger than order closedness). In a second step, under convexity, we establish the equivalence between order closedness and closedness with respect to a convenient weak topology (which fails to be locally solid in general). This result will be key to derive powerful dual representations for surplus-invariant risk functionals in the sequel. Here, we denote by $\cX^\sim_n$ the order-continuous dual of $\cX$. We say that a family $\cI\subset\cX^\sim_n$ is {\em separating} if for every distinct elements $X,Y\in\cX$ there exists $\varphi\in\cI$ such that $\varphi(X)\neq\varphi(Y)$. Recall that $\sigma(\cX,\cI)$ is Hausdorff precisely when $\cI$ is separating.

\begin{theorem}
\label{closed-top}
For every surplus-invariant monotone set $\cA\subset\cX$ the following statements are
equivalent:
\begin{enumerate}
  \item[(a)] $\cA$ is order closed.
  \item[(b)] $\cA$ is uo closed.
\end{enumerate}
If $\cX^\sim_n$ is separating and $\cA$ is convex, the above are also equivalent to:
\begin{enumerate}
  \item[(c)] $\cA$ is $\sigma(\cX,\cI)$ closed for every separating ideal $\cI\subset\cX^\sim_n$.
  \item[(d)] $\cA$ is $\sigma(\cX,\cI)$ closed for some separating ideal $\cI\subset\cX^\sim_n$.
\end{enumerate}
\end{theorem}
\begin{proof}
Throughout the proof we set $\cD=-\cA_-$. Recall from Proposition~\ref{s-i} that $\cD$ is solid in $\cX_+$. Since order convergence always implies unbounded-order convergence, we
immediately see that {\em (b)} implies {\em (a)}. To prove the converse
implication, assume that $\cA$ is order closed so that $\cD$ is also order closed by Proposition~\ref{prop: closedness}. By applying
point {\em (b)} in Lemma~\ref{con} to the solid set $\cD$ we easily conclude that $\cD$ is unbounded-order closed. Consequently, $\cA$
is unbounded-order
closed as well by virtue of Proposition~\ref{prop: closedness}, proving that {\em (a)} implies {\em (b)}.

\smallskip

Now, assume that $\cX^\sim_n$ is separating and $\cA$ is convex. It is clear that {\em (c)} implies {\em (d)}. Next, assume that $\cI$ is a separating
ideal of $\cX^\sim_n$. Since $\sigma(\cX,\cI)$ is clearly order continuous, it is immediate to see that order convergence implies
$\sigma(\cX,\cI)$ convergence and, hence, {\em (d)} implies {\em (a)}. To conclude the proof, assume that $\cA$ is order closed so
that $\cD$ is itself order closed by Proposition~\ref{prop: closedness}. Note that
\[
\cS = \{X\in\cX \,; \ \abs{X}\in \cD\}
\]
is a solid set in $\cX$. Since lattice operations are order continuous, it is immediate to see that $\cS$ is also order closed. Recall
that $\abs{\sigma}(\cX,\cI)$ is locally solid and order continuous. Then, we can apply Lemma~\ref{lem: amemiya-fremlin} to
infer that $\cS$ is closed with respect to the absolute topology
$\abs{\sigma}(\cX,\cI)$. In particular, since $\cD=\cS\cap \cX_+$, this implies
that $\cD$, hence $\cA$ by Proposition~\ref{prop: closedness}, is also
$\abs{\sigma}(\cX,\cI)$ closed. Recall that $\cI$ is the topological dual of
$\cX$ equipped with the topology $\abs{\sigma}(\cX,\cI)$ by Theorem~2.33 in Aliprantis and Burkinshaw
(2003).  Being convex, the set $\cA$ is then automatically $\sigma(\cX,\cI)$
closed. This establishes that {\em (a)} implies {\em (c)} and concludes the
proof of the equivalence.
\end{proof}

\smallskip

\begin{remark}
(i) In general, the smaller the ideal $\cI$ is, the stronger the above point {\em (d)}. Under the assumptions of the theorem, point
{\em (c)} tells us that all the choices of $\cI$ have the same strength. Thus, one has the freedom to choose $\cI$ as small as possible, or as
convenient as needed.

\smallskip

(ii) The reason why, in the context of the preceding result, we did not apply Lemma~\ref{lem: amemiya-fremlin} to the topology $\sigma(\cX,\cI)$ to immediately establish the desired equivalence is that $\sigma(\cX,\cI)$ typically fails to be locally solid.
%Indeed, as implied by Theorem~2.36 in Aliprantis and Burkinshaw (2003), the topology $\sigma(\cX,\cI)$
%can be locally solid only when each principal ideal of $\cI$ is of finite dimension. To see how restrictive this is, consider a probability space $(\Omega,\cF,\probp)$ and let $\cX$ be an ideal of $L^1(\Omega,\cF,\probp)$. Then, as soon as $\cI$ contains a constant random variable, the local solidity of $\sigma(\cX,\cI)$ would force $L^\infty(\Omega,\cF,\probp)$ to have finite dimension.
\end{remark}

\smallskip

\begin{remark}[On solid subsets of $\cX_+$]
Theorem~\ref{closed-top} remains valid if we replace the convex surplus-invariant monotone set $\cA\subset\cX$ with a convex set $\cC\subset\cX_+$ that is solid in $\cX_+$. In view of Proposition~\ref{s-i}, this follows at once from Corollary~\ref{cor: convexity}
and Proposition~\ref{prop: closedness} by considering $\cA=\cX_+-\cC$.
\end{remark}

\smallskip

%%%%%%%%%%%%%%%%%%%%%%%%%%%%%%%%%%%%%%%%%%%%%%%%%

\subsection{A decomposition result for surplus-invariant monotone sets}

We now turn to the announced decomposition result for convex, surplus invariant,
monotone sets. In order to streamline the proof, it is useful to first state the
following simple result about convex sets in a vector lattice.

\begin{lemma}
\label{add}
For every convex order-closed set $\cA\subset\cX$ containing a cone $\cC$ we have
$\cA+\cC\subset\cA$.
%Consider a convex set $\cA\subset\cX$.
%If $\cA$ contains two cones $\cC_1$ and
%$\cC_2$, then we have $\cC_1+\cC_2\subset\cA$.
%If $\cA$ is order closed and
%contains a cone $\cC$, then we have $\cA+\cC\subset\cA$.
\end{lemma}
\begin{proof}
%The first statement follows immediately once we note that
%$X_1+X_2=\frac{1}{2}(2X_1)+\frac{1}{2}(2X_2)$ for all $X_1
%\in\cC_1$ and $X_2\in\cC_2$.
Take arbitrary $X\in\cA$ and $Y\in\cC$ and note that
$Y+(1-\frac{1}{n})X\xrightarrow{o}Y+X$. Since for all $n\in\N$ we have
$nY\in\cC$ by assumption, hence $nY\in\cA$, and
\[
Y+\Big(1-\frac{1}{n}\Big)X = \frac{1}{n}(nY)+\Big(1-\frac{1}{n}\Big)X \in \cA,
\]
we infer from the order closedness of $\cA$ that $X+Y\in\cA$, concluding the
proof.
\end{proof}

\smallskip

We can now state and prove our decomposition result. Here, given a convex set
$\cA\subset\cX$, we denote by $\rec(\cA)$ its {\em recession cone}, i.e.~we set
\[
\rec(\cA) = \bigcap_{\lambda>0}\lambda(\cA-Z).
\]
for a fixed $Z\in\cA$ (the above intersection is independent of the choice of $Z\in\cA$).
Moreover, we denote by $\lin(\cA)$ its {\em lineality space}, i.e.~we set
\[
\lin(\cA) = \rec(\cA)\cap(-\rec(\cA)).
\]
Note that, if $\cA$ contains the zero element (as is the case under monotonicity
and surplus invariance), then $\rec(\cA)$ coincides with the largest convex cone
contained in $\cA$. Similarly, $\lin(\cA)$ is the largest vector space contained
in $\cA$. Moreover, we say that $\cA$ is {\em radially
bounded} if for every $X\in\cA\backslash\{0\}$ there exists $\lambda_X>0$ such
that $\lambda X\notin\cA$ for all $\lambda\geq\lambda_X$.

\begin{theorem}
\label{decom}
Assume $\cX $ has the projection property. Then, every convex order-closed surplus-invariant monotone set $\cA\subset\cX$ admits a
unique decomposition
\begin{equation}
\label{eq: main decomposition}
\cA = (\cB_1)_+\oplus\big((\cB_2)_+-\cD\big)\oplus\cB_3
\end{equation}
where $\cB_1$, $\cB_2$, $\cB_3$ are bands such that $\cX=\cB_1\oplus\cB_2\oplus\cB_3$ and $\cD$ is a subset of $(\cB_2)_+$ satisfying
the following properties:
\begin{enumerate}
  \item[(a)] $\cD$ is convex, order closed, radially bounded and solid in
$\cX_+$.
  \item[(b)] $\cD=-\cA_-\cap\cB_2$.
  \item[(c)] $X\in(\cB_2)_+\backslash\{0\}$ implies $X\wedge Y\neq0$ for some
$Y\in\cD$.
  \item[(d)] $\cB_2=\band(\cD)$.
\end{enumerate}
In this case, we have $\rec(\cA)=(\cB_1)_+\oplus(\cB_2)_+\oplus\cB_3$ and $\lin(\cA)=\cB_3$ and we say that $\cA$ is {\em represented
by $(\cB_1;\cB_2,\cD;\cB_3)$}.
%the following statements hold:
%\begin{enumerate}
%	\item[(e)] $\rec(\cA)=(\cB_1)_+\oplus(\cB_2)_+\oplus\cB_3$.
%	\item[(f)] $\lin(\cA)=\cB_3$.
%\end{enumerate}
%If $\cA$ admits the decomposition \eqref{eq: main decomposition}, we say that
%$\cA$ is {\em represented by $(\cB_1;\cB_2,\cD;\cB_3)$}.
\end{theorem}
\begin{proof}
First of all, define the bands
\[
\cB_1=(\cA_-)^d, \ \ \ \cB_3=\band(\rec(\cA_-)), \ \ \
\cB_2=(\cB_1\oplus\cB_3)^d.
\]
Note that, by definition of $\cB_2$, we clearly have $\cX=\cB_1\oplus\cB_2\oplus\cB_3$. Moreover, consider the subset of $(\cB_2)_+$
given by
\[
\cD = -\cA_-\cap\cB_2.
\]
Clearly, $\cD$ satisfies {\em (b)}. Moreover, it is easily seen to be convex,
order closed, and solid in $\cX_+$. To show that $\cD$ is also radially bounded, assume to the contrary that
$\lambda_nX\in\cD$ for some strictly-increasing sequence
$\lambda_n\uparrow\infty$ and some $X\in\cD\backslash\{0\}$. In this case, we
would get $-X\in\rec(\cA_-)$ by solidity and thus $X\in\cB_3$. However, since
we assumed that $X\in\cB_2$, this would imply $X=0$ in contrast to our initial
assumption. As a result, $\cD$ must be radially bounded. This establishes {\em
(a)}.

\smallskip

To show {\em (c)}, assume there exists some $X\in(\cB_2)_+\backslash\{0\}$ such
that $X\wedge Y=0$ for all $Y\in\cD$. For each $i\in\{1,2,3\}$, let $P_i$ be the
band projection of $\cB_i$. By Proposition~\ref{prop: si and projections}, one easily sees that
\begin{equation}
\label{pro-ima}
P_i(\cA)=\cA\cap\cB_i \ \ \ \mbox{and} \ \ \ P_i(\cA_-) = \cA_-\cap\cB_i
\end{equation}
%so that
%\begin{equation}
%\label{pro-ima-n}
%P_i(\cA_-) = \cA_-\cap\cB_i
%\end{equation}
for all $i\in\{1,2,3\}$. Now, take an arbitrary $Z\in-\cA_-$ and note that,
since $X\in\cB_2$, we clearly have $X\wedge P_1(Z)=X\wedge P_3(Z)=0$. Moreover,
we have $X\wedge P_2(Z)=0$ by definition of $\cD$. This implies that
$X\wedge Z=0$. Since $Z$ was chosen arbitrarily in $-\cA_-$, it follows that
$X\in\cB_1$. However, this is not possible because $X$ is a nonzero element
belonging to $\cB_2$. This proves that {\em (c)} must hold. Since
$\cB_2=\band(\cD)$ is equivalent to {\em (c)} by Theorem~1.39 in Aliprantis and
Burkinshaw (2006) applied to $\cD$ in $\cB_2$, we conclude that {\em (d)} holds
as well.

\smallskip

We proceed to establish the decomposition~\eqref{eq: main decomposition}. First,
take $X\in\cA\cap\cB_1$ and note that
\[
-X^- \in \cA_-\cap\cB_1 = \cA_-\cap(\cA_-)^d = \{0\}.
\]
This shows that $\cA\cap\cB_1\subset(\cB_1)_+$. Since we can write
\[
X = P_1(X)+(P_2(X^+)-P_2(X^-))+P_3(X)
\]
for all $X\in\cX$, we infer from \eqref{pro-ima} that
\[
\cA \subset (\mathcal{B}_1)_+\oplus\big((\cB_2)_+-\cD\big)\oplus\cB_3.
\]
To prove the converse inclusion, note first that $\cA$ contains $-\cD$ and thus,
by monotonicity, also $(\cB_1)_+\oplus\big((\cB_2)_+-\cD\big)$. We claim that
$\cA$ contains $\cB_3$ as well. To show this, take any $n\in\N$, any
$\lambda_1,\dots,\lambda_n>0$, any $X_1,\dots,X_n\in\rec(\cA_-)$ and any
$Y\in\cX$ such that
$\abs{Y}\leq\sum_{i=1}^n\lambda_i\abs{X_i}=-\sum_{i=1}^n\lambda_iX_i$. Note that
\[
Y \geq \sum_{i=1}^n\lambda_iX_i =
\sum_{i=1}^n\frac{\lambda_i}{\sum_{j=1}^n\lambda_j}\Big(\sum_{j=1}
^n\lambda_j\Big)X_i \in \cA
\]
by convexity of $\cA$, showing that $Y\in\cA$ by monotonicity. This implies that
$\cA$ contains the
ideal generated by $\rec(\cA_-)$ and, being order closed, also the band
generated by $\rec(\cA_-)$ or, equivalently, $\cB_3$. As a result, it follows
from Lemma~\ref{add} that
\[
\cA \supset \cA+\cB_3 \supset (\cB_1)_+\oplus\big((\cB_2)_+-\cD\big)\oplus\cB_3.
\]
The decomposition~\eqref{eq: main decomposition} is then established.

\smallskip

To prove that $\rec(\cA)=(\cB_1)_+\oplus(\cB_2)_+\oplus\cB_3$ and thus $\lin(\cA)=\cB_3$, note first that $\rec(\cA)$ clearly contains
$(\cB_1)_+\oplus(\cB_2)_+\oplus\cB_3$. To show the converse inclusion,
it suffices to prove that $\rec(\cA\cap\cB_2)\subset\cX_+$. To this effect, take
$X\in\rec(\cA\cap\cB_2)$ and note that $-X^-\in\rec(\cA\cap\cB_2)$ as well by
surplus invariance. In particular, $X^-$ belongs to $\rec(\cD)$. Since $\cD$ is
radially bounded, we must have $\rec(\cD)=\{0\}$ and this shows that
$X\in\cX_+$.

\smallskip

Finally, to prove the uniqueness of the above decomposition, assume that
\[
\cA = (\cB'_1)_+\oplus\big((\cB'_2)_+-\cD'\big)\oplus\cB'_3
\]
where $\cB'_1$, $\cB'_2$, $\cB'_3$ are bands such that
$\cX=\cB'_1\oplus\cB'_2\oplus\cB'_3$ and $\cD'$ is a subset of $(\cB'_2)_+$
satisfying conditions {\em (a)} and {\em (d)} (once we
replace $\cD$ by $\cD'$ and $\cB_2$ by $\cB'_2$). Clearly,
$\cB'_1=(\cA_-)^d=\cB_1$  and $\cB'_3=\lin(\cA)=\cB_3$. This implies that
$\cB'_2=\cB_2$. By Proposition~\ref{s-i}, it is also easy to show that
$\cA_-=-\cD'\oplus (\cB_3')_-$, so that $\cD'=-\cA_-\cap \cB_2'$. This implies
$\cD'=\cD$, concluding the proof of uniqueness.
\end{proof}

\smallskip

The next corollary specifies the above decomposition to the case of conic surplus-invariant sets.

\begin{corollary}
Assume $\cX$ has the projection property. Then, every convex order-closed surplus-invariant monotone set $\cA\subset\cX$ that is additionally a cone admits a unique decomposition
\[
\cA = \cB_+\oplus\cB^d
\]
where $\cB$ is a band.
\end{corollary}
\begin{proof}
Assume that $\cA$ is represented by $(\cB_1;\cB_2,\cD;\cB_3)$. Since
$\cA=\rec(\cA)$, it follows from Theorem \ref{decom} that $\cA = (\cB_1)_+\oplus(\cB_2)_+\oplus\cB_3$. Setting $\cB=\cB_1\oplus\cB_2$
immediately yields the desired claim.
\end{proof}

\smallskip

\begin{remark}[On solid subsets of $\cX_+$]
In view of Proposition~\ref{s-i}, Corollary~\ref{cor: convexity}, and Proposition~\ref{prop: closedness}, the above decomposition result can be easily adapted to every convex order-closed set $\cC\subset\cX_+$ that is solid in $\cX_+$. In this case, provided that $\cX$ has the projection property, $\cC$ admits a unique decomposition $\cC = \cD\oplus\cB_+$ where $\cB$ is a band and $\cD$ is a convex order-closed radially-bounded subset of $\cB^d$ that is solid in $\cX_+$. In this case, we say that $\cC$ is {\em
represented by $(\cD;\cB)$}. If $\cC$ is additionally assumed to be a cone, then we must have $\cC = \cB_+$ for a suitable band $\cB$.
\end{remark}

\smallskip

%%%%%%%%%%%%%%%%%%%%%%%%%%%%%%%%%%%%%%%%%%%%

\subsection{Strictly-positive supporting functionals}

We apply the above closedness and decomposition results to investigate the existence of strictly-positive and order-continuous functionals that support a convex surplus-invariant monotone set. We start with the following preliminary boundedness result.

\begin{lemma}
\label{lem: strictly positive supporting functional}
Assume $\cX_n^\sim$ has the countable sup property and admits a strictly-positive element. Let $\cC\subset\cX_+$ be a convex order-closed radially-bounded solid subset of $\cX_+$. Then, there exists a strictly-positive functional $\varphi\in\cX^\sim_n$ such
that $\sup_{X\in\cC}\varphi(X)<\infty$.
\end{lemma}
\begin{proof}
Clearly, we can assume that $\cC$ contains nonzero elements. Fix a nonzero $X\in\cC$ and note that $\lambda_X X\notin\cC$ for a
suitable $\lambda_X>0$. Since $\cX^\sim_n$ admits a strictly-positive element and is thus separating, it follows from
Theorem~\ref{closed-top} that $\cC$ is $\sigma(\cX,\cX_n^\sim)$ closed. As a consequence of Hahn-Banach Separation, there
exists a functional $\varphi_X\in\cX_n^\sim$ satisfying
\begin{equation}
\label{eq: strictly positive support functional 1}
\sup_{Y\in\cC}\varphi_X(Y) < \lambda_X\varphi_X(X).
\end{equation}
Since $\cC$ is solid in $\cX_+$, the above supremum coincides with $\sup_{Y\in\cC}{\varphi_X^+}(Y)$ by the Riesz-Kantorovich formula,
see Theorem~1.18 in Aliprantis and Burkinshaw~(2006). As a result, we obtain
\[
\sup_{Y\in \cC}\varphi_X^+(Y) < \lambda_X\varphi_X(X) \leq \lambda_X{\varphi_X^+}(X),
\]
where we used that $X\in\cX_+$. Hence, we may assume without loss of generality that $\varphi_X$ is a positive functional. Now, take a
strictly-positive functional $\psi\in\cX_n^\sim$,
which exists by assumption, and define
\[
\varphi = \sup_{X\in\cC}(\varphi_X\wedge\psi).
\]
The supremum exists in $\cX_n^\sim$ because $\cX_n^\sim$ is always order
complete. By the countable sup property of $\cX_n^\sim$, we find a sequence
$(X_n)\subset\cC$ such that
\begin{equation}
\label{eq: strictly positive support functional 3}
\varphi = \sup_{n\in\N}(\varphi_{X_n}\wedge\psi).
\end{equation}
Moreover, choose a sequence of strictly-positive real numbers
$(\lambda_n)\subset\R$ such that
\[
\sum_{n\in\N}\lambda_n\big(\sup_{Y\in\cC}\varphi_{X_n}(Y)\big) \in \R
\]
and define
\[
\varphi_0 = \sup_{n\in\N}\big((\lambda_n\varphi_{X_n})\wedge\psi\big).
\]
Note that $\sup_{X\in\cC}\varphi_0(X)<\infty$ because
\[
\varphi_0(X) = \lim_{k\to\infty}\sup_{n\leq k}\big((\lambda_n\varphi_{X_n})\wedge\psi\big)(X) \leq
\sum_{n\in\N}\lambda_n\varphi_{X_n}(X) \leq
\sum_{n\in\N}\lambda_n\bigg(\sup_{Y\in\cC}\varphi_{X_n}(Y)\bigg)
\]
for every $X\in\cC$.

\smallskip

We show that $\varphi_0$ is strictly positive on $\cC^{dd}$. To this effect, assume to the contrary that we find a nonzero positive
$U\in C^{dd}$ such that $\varphi_0(U)=0$. Then, we must have
\[
0 \leq \min(\lambda_n,1)(\varphi_{X_n}\wedge\psi)(U) \leq \big((\lambda_n\varphi_{X_n})\wedge\psi\big)(U) = 0
\]
for every $n\in\N$. In particular, we get $(\varphi_{X_n}\wedge\psi)(U)=0$ for all $n\in\N$. Fix now $n\in\N$ and set
\[
\varphi_k = \sup_{n\leq k}(\varphi_{X_n}\wedge\psi)
\]
for $k\in\N$. For every $k\in\N$ we can use the inequality $\varphi_k \leq \sum_{n\leq k}\varphi_{X_n} \wedge \psi$ to obtain that
$\varphi_k(U)=0$. Since $(\varphi_k)$ is an increasing sequence that is order convergent to $\varphi$ in $\cX_n^\sim$, we can apply
Theorem~1.18 in Aliprantis and Burkinshaw (2006) to infer that $\varphi(U)=\lim_{k\to\infty}\varphi_k(U) = 0$. This implies
\[
(\varphi_X\wedge\psi)(U)=0 \ \ \ \mbox{for all $X\in C$}.
\]
Pick any $X\in C$. We claim that $\varphi_X(U)=0$. To see this, take $\varepsilon>0$ and use Theorem~1.18 in Aliprantis and Burkinshaw
(2006) to find for each $n\in\N$ an element $U_n\in\cX_+$ satisfying $U_n\leq U$ and
\[
0 \leq \varphi_X(U-U_n)+\psi (U_n)\leq \frac{\varepsilon}{2^n}.
\]
Set $V_n=\bigwedge_{m=1}^n U_m$ for $n\in\N$ and note that $(V_n)$ is decreasing and $\psi(V_n)\to0$. Then, by strict positivity of
$\psi$, it is easy to see that $\inf_{n\in\N}V_n=0$, so that $U-V_n\xrightarrow{o}U$ and $\varphi_X(U-V_n)\to\varphi_X(U)$. But for
every $n\in\N$ we have
\[
\varphi_X(U-V_n)=\varphi_X\Big(\bigvee_{m=1}^n(U-U_m)\Big)\leq \sum_{m=1}^n\varphi_X(U-U_m)\leq \varepsilon,
\]
implying $\varphi_X(U)=0$. This proves the claim. Now, take some element $V\in\cC$ such that $U\wedge V$ is nonzero. Such an element
exists for otherwise $U$ would belong to $\cC^d$ and $\cC^{dd}$ by assumption, leading to the contradiction $U=0$. Then, $U\wedge
V\in\cC$ by solidity of $\cC$. As a result, the preceding claim and the positivity of $\varphi_{U\wedge V}$ yield
\[
0 = \varphi_{U\wedge V}(U) \geq \varphi_{U\wedge V}(U\wedge V) \geq 0
\]
so that $\varphi_{U\wedge V}(U\wedge V)=0$. However, since $0\in\cC$, this contradicts~\eqref{eq: strictly positive support functional
1}. In conclusion, it follows that $\varphi_0$ is strictly positive on $\cC^{dd}$.

\smallskip

Finally, set $\varphi_1(X)=\sup\{\psi(X\wedge Y) \,; \ Y\in\cC^d\cap\cX_+\}$ for $X\in\cX_+$. Then, a similar argument as in
Theorem~1.22 in Aliprantis and Burkinshaw~(2006) shows that $\varphi_1$ extends to a unique positive functional on the entire $\cX$.
Clearly, $\varphi_1\leq \psi$ so that $\varphi_1\in\cX_n^\sim$. Moreover, $\varphi_1=0$ on $\cC^{dd}$ and $\varphi_1=\psi$ on $\cC^d$.
Now, define $\varphi^\ast=\varphi_0+\varphi_1$. Then, $\varphi^\ast\in\cX_n^\sim$ and, since $\cC\subset \cC^{dd}$,
\[
\sup_{X\in\cC}\varphi^\ast(X) = \sup_{X\in\cC}\varphi_0(X) < \infty.
\]
To conclude the proof, it suffices to prove that $\varphi^\ast$ is strictly positive on $\cX$. To this effect, take a nonzero
$X\in\cX_+$ and recall from Theorem~1.36 in Aliprantis and Burkinshaw~(2006) that there exist positive $U\in\cC^d$ and $V\in\cC^{dd}$
such that $0<U+V\leq X$. This yields $\varphi^\ast(X)\geq\varphi^\ast(U+V)\geq\varphi_1(U)+\varphi_0(V)>0$ by the strict positivity of
$\varphi_1$ on $\cC^d$ and of $\varphi_0$ on $\cC^{dd}$.
\end{proof}
%\begin{remark}
%Incidentally, note that one could use the above result to deduce the fundamental fact that every Banach function space over a $\sigma$-finite measure space admits a strictly-positive order continuous functional.
%\end{remark}

\smallskip

Here, we define the {\em order-continuous barrier cone} of a set $\cA\subset\cX$ by setting
\[
\barr(\cA) = \big\{\varphi\in\cX_n^\sim \,; \
\inf_{X\in\cA}\varphi(X)>-\infty\big\}.
\]
The next result shows that, under suitable assumptions on the underlying space, the ``nontrivial'' part of a convex order-closed surplus-invariant monotone set can be supported by a strictly-positive order-continuous linear functional.

\begin{theorem}
\label{theo: strictly positive supporting functional}
Assume $\cX$ has the projection property and $\cX_n^\sim$ has the countable
sup property and admits a strictly-positive element. Consider a convex order-closed surplus-invariant monotone set $\cA\subset\cX$ represented by $(\cB_1;\cB_2,\cD;\cB_3)$. Then, there exists $\varphi\in\barr(\cA)$ that is strictly positive on $\cB_1\oplus \cB_2$.
\end{theorem}
\begin{proof}
Recall from Theorem \ref{decom} that $\cA$ can be decomposed as
\[
\cA = (\cB_1)_+\oplus\big((\cB_2)_+-\cD\big)\oplus\cB_3,
\]
where $\cD=-\cA_-\cap\cB_2$ is convex, order closed, radially bounded and solid
in $\cX_+$. Note that every functional in $\barr(\cA)$ must annihilate the vector space $\cB_3$. The desired functional can be
obtained by applying Lemma~\ref{lem: strictly positive supporting functional} and by setting the corresponding functional equal to
zero on $\cB_3$.
\end{proof}

\smallskip

%%%%%%%%%%%%%%%%%%%%%%%%%%%%%%%%%%%%%%%%

\section{Surplus-invariant risk functionals}
\label{sec: risk measures}

In this section we discuss two possible ways to transfer the notion of surplus invariance to the realm of functionals. Our main focus is on dual representations and extensions of surplus-invariant functionals. Throughout this section we maintain our assumption that $\cX$ is a fixed vector lattice.

\begin{definition}
A functional $\rho:\cX\to\overline{\R}$ is {\em monotone} if for every $X,Y\in\cX$ we have
\[
X\geq Y \ \implies \ \rho(X)\leq\rho(Y).
\]
We say that $\rho$ is {\em quasiconvex} if for every $X,Y\in\cX$ and
$\lambda\in[0,1]$ we have
\[
\rho(\lambda X+(1-\lambda)Y) \leq \max(\rho(X),\rho(Y)).
\]
We say that $\rho$ is {\em $S$-additive} for $S\in\cX_+\backslash\{0\}$ if for every $X\in\cX$ and $m\in\R$ we have
\[
\rho(X+mS) = \rho(X)-m.
\]
We say that $\rho$ satisfies the {\em Fatou property} if it is lower semicontinuous with respect to order convergence, i.e.~for every
net $(X_\alpha)\subset\cX$ and $X\in\cX$ we have
\[
X_\alpha\xrightarrow{o}X \ \implies \ \rho(X)\leq\liminf_{\alpha}\rho(X_\alpha).
\]
Similarly, $\rho$ is said to satisfy the {\em super Fatou property} if it is
lower semicontinuous with respect to unbounded-order convergence, i.e.~for every $(X_\alpha)\subset\cX$ and $X\in\cX$ we have
\[
X_\alpha\xrightarrow{uo}X \ \implies \ \rho(X)\leq\liminf_{\alpha}\rho(X_\alpha).
\]
\end{definition}

\smallskip

The property of monotonicity is standard in risk measure theory, see F\"{o}llmer and Schied~(2017). The notion of quasiconvexity has
also been extensively studied in a risk measure context, see e.g.~Cerreia-Vioglio et al.~(2011) and Drapeau and Kupper (2013). The
property of $S$-additivity was introduced in Artzner et al.~(1999) and is crucial to assign an operational interpretation to a risk
functional. Indeed, in this case,
\begin{equation}
\label{eq: S additivity}
\rho(X) = \inf\{m\in\R \,; \ X+mS\in\cA\}
\end{equation}
for every position $X\in\cX$, where $\cA=\{\rho\leq0\}$. Hence, the quantity $\rho(X)$ can be naturally interpreted as the smallest
number of units of the ``reference asset'' with payoff $S$ that has to be acquired in order to ensure ``acceptability''. We refer to
Munari (2015) for a comprehensive treatment of this type of risk functionals.

\smallskip

The Fatou property plays a critical role in obtaining tractable dual representations for convex risk measures. The original
contribution, in the setting of bounded random variables, goes back to Delbaen (2002). Note that, if $\cX$ has the countable sup
property as it is the case whenever $\cX$ is an ideal of $L^0(\Omega,\cF,\probp)$, then one can replace the net $(X_\alpha)$ with a
sequence in the definition of the (super) Fatou property. Indeed, by the countable sup property, if a net $(X_\alpha)$
(unbounded-)order converges to $X$, then we find a countable family $(\alpha_n)$ such that $(X_{\alpha_n})$ (unbounded-)order
converges to $X$ as well.

\smallskip

A simple characterization of the above properties in terms of sublevel sets is recorded in the next result. The simple proof is
omitted.

\begin{lemma}
\label{lem: properties rho}
For a functional $\rho:\cX\to\overline{\R}$ and $S\in\cX_+\backslash\{0\}$ the following hold:
\begin{enumerate}
  \item[(a)] $\rho$ is monotone iff $\{\rho\leq m\}$ is monotone for each $m\in\R$.
  \item[(b)] $\rho$ is quasiconvex iff $\{\rho\leq m\}$ is convex for each $m\in\R$.
  \item[(c)] $\rho$ is $S$-additive iff $\{\rho\leq m\}=\{\rho\leq0\}-mS$ for each $m\in\R$.
  \item[(d)] $\rho$ has the Fatou property iff $\{\rho\leq m\}$ is order closed for each $m\in\R$.
  \item[(e)] $\rho$ has the super Fatou property iff $\{\rho\leq m\}$ is uo closed for each $m\in\R$.
\end{enumerate}
\end{lemma}

\smallskip

%%%%%%%%%%%%%%%%%%%%%%%%%%%%%%%%%%%%%%%%%%%

\subsection{Surplus invariance}

We start by focusing on surplus-invariant monotone maps.

\begin{definition}
A functional $\rho:\cX\to\overline{\R}$ is {\em surplus invariant} if for every $X\in\cX$
\[
\rho(X) = \rho(-X^-).
\]
\end{definition}

\smallskip

The notion of surplus invariance for risk functionals has been studied under the name of {\em loss-based property} in Cont et
al.~(2013) and of {\em excess invariance} in Staum (2013) in the context of spaces of bounded random variables. Interestingly enough,
surplus invariance was already listed among the defining properties of a coherent risk measure in a Technical Report of Cornell
University by Artzner et al.~(1996), where it was stated that a good solvency metric, in line with a long-rooted actuarial tradition,
should be indifferent to the size of profits and only focus on losses. We will see later why surplus invariance eventually disappeared
from the axioms of coherence in Artzner et al.~(1999).

\smallskip

The next simple result shows that a functional is surplus invariant precisely when each of its sublevel sets is surplus invariant.

\begin{proposition}
For a functional $\rho:\cX\to\overline{\R}$ the following are equivalent:
\begin{enumerate}
	\item[(a)] $\rho$ is surplus invariant.
	\item[(b)] $\{\rho\leq m\}$ is surplus invariant for each $m\in\R$.
\end{enumerate}
\end{proposition}
\begin{proof}
It is clear that {\em (a)} implies {\em (b)}. To show the converse implication, assume that {\em (b)} holds and note that $\rho(X) =
\inf\{m\in\R \,; \ X\in\{\rho\leq m\}\}$ for every $X\in\cX$. Since $\{\rho\leq m\}$ is surplus invariant for each $m\in\R$, we easily
see that $\rho(X)=\rho(-X^-)$ for every $X\in\cX$ and therefore {\em (a)} holds.
\end{proof}

\smallskip

The next result provides a characterization of the Fatou property for surplus-invariant monotone functionals that are quasiconvex. The
statement is a direct consequence of Theorem~\ref{closed-top}.

\begin{theorem}
\label{s-i functionals}
For a quasiconvex surplus-invariant monotone functional $\rho:\cX\to\R\cup\{\infty\}$ the
following statements are equivalent:
\begin{enumerate}
  \item[(a)] $\rho$ has the Fatou property.
  \item[(b)] $\rho$ has the super Fatou property.
\end{enumerate}
If $\cX^\sim_n$ is separating, then the above are also equivalent to:
\begin{enumerate}
  \item[(c)] $\rho$ is $\sigma(\cX,\cI)$ lower semicontinuous for every separating ideal $\cI\subset\cX^\sim_n$.
  \item[(d)] $\rho$ is $\sigma(\cX,\cI)$ lower semicontinuous for some separating ideal $\cI\subset\cX^\sim_n$.
\end{enumerate}
\end{theorem}

\smallskip

\begin{remark}
\label{rem: dual representation}
(i) In general, the smaller the ideal $\cI$, the stronger the statement in point {\em (d)}. Under the assumptions of the theorem,
point {\em (c)} says that all the choices of $\cI$ have the same strength. Thus, one can choose $\cI$ as small as possible, or as
convenient as needed.

\smallskip

(ii) It follows from classical convex duality that, if $\rho$ is additionally convex instead of being just quasiconvex, then {\em (d)}
is equivalent to
\begin{equation}
\label{eq: dual representation}
\rho(X) = \sup_{\varphi\in\cI}\big(\varphi(X)-\rho^\ast(\varphi)\big)
\end{equation}
for every $X\in\cX$, where
\[
\rho^\ast(\varphi) = \sup_{X\in\cX}\big(\varphi(X)-\rho(X)\big).
\]
Note that, since $\rho$ is monotone, $\rho^\ast(\varphi)<\infty$ implies that $\varphi$ is negative. Thus, we can restrict the
supremum in the representation of $\rho$ to the set $\cI_-$. Moreover, in view of surplus invariance, we can equivalently write
\[
\rho(X) = \sup_{\varphi\in\cI_-}\big(\varphi(-X^-)-\rho^\ast(\varphi)\big)
\]
for every $X\in\cX$. In general, it is desirable to choose a small set of dual elements $\cI$ to yield a more tractable dual
representation. We will demonstrate the power of this result in the next section in both classical and robust models of financial
positions.
\end{remark}

\smallskip

We proceed to discussing extension properties of surplus-invariant functionals. In particular, we show that every functional defined
on the principal ideal generated by a weak order unit in $\cX$ can be always extended to the whole of $\cX$ by preserving surplus
invariance, monotonicity, and the (super) Fatou property. An element $U\in\cX_+$ is said to be a {\em weak order unit} if
$\abs{X}\wedge nU\uparrow\abs{X}$ for every $X\in\cX$. The smallest ideal containing $U$ is denoted by $\cI_U$.

\begin{lemma}
\label{lem: unique surplus invariant extension}
Let $U$ be a weak order unit of $\cX$ and consider two surplus-invariant monotone functionals with the Fatou property
$\rho_1,\rho_2:\cX\to\R\cup\{\infty\}$. Then, $\rho_1=\rho_2$ on $\cI_U$ implies $\rho_1=\rho_2$ on $\cX$.
\end{lemma}
\begin{proof}
Fix $X\in\cX$. Since $(X^-\wedge nU)\leq X^-$ for every $n\in\N$ and $X^-\wedge nU\xrightarrow{o}X^-$, we have
\[
\limsup_{n\to\infty}\rho_1(-X^-\wedge nU) \leq \rho_1(-X^-) \leq \liminf_{n\to\infty}\rho_1(-X^-\wedge nU).
\]
As a result, surplus invariance yields
\begin{equation}
\label{eq: extension}
\rho_1(X) = \lim_{n\to\infty}\rho_1(-X^-\wedge nU).
\end{equation}
Clearly, the same holds for $\rho_2$. Since $-X^-\wedge nU\in\cI_U$ for all $n\in\N$ and $\rho_1$ coincides with $\rho_2$ on the
principal ideal $\cI_U$, we infer that $\rho_1(X)=\rho_2(X)$.
\end{proof}

\smallskip

\begin{theorem}
\label{theo: general extension}
Let $U$ be a weak order unit of $\cX$. Then, every surplus-invariant monotone functional with the Fatou property
$\rho:\cI_U\to\R\cup\{\infty\}$ can be uniquely extended to a surplus-invariant monotone functional with the Fatou property
$\overline{\rho}:\cX\to\R\cup\{\infty\}$. In particular, for every $X\in\cX$ we have
\[
\overline{\rho}(X) = \lim_{n\to\infty}\rho(-X^-\wedge nU).
\]
In addition, if $\rho$ is (quasi)convex, then $\overline{\rho}$ is (quasi)convex as well.
\end{theorem}
\begin{proof}
The functional $\overline{\rho}:\cX\to\R\cup\{\infty\}$ is well-defined because $-X^-\wedge nU\in\cI_U$ for all $X\in\cX$ and $n\in\N$
and $\rho(-X^-\wedge nU)\uparrow\overline{\rho}(X)$ for every $X\in\cX$. In view of \eqref{eq: extension} applied to $\rho$, it
follows that $\overline{\rho}$ is an extension of $\rho$. It is immediate to see that $\overline{\rho}$ is surplus invariant and
monotone. To show that $\overline{\rho}$ has the Fatou property, consider for every $n\in\N$ the functional
$\rho_n:\cX\to\R\cup\{\infty\}$ defined by
\[
\rho_n(X) = \rho(-X^-\wedge nU)
\]
for $X\in\cX$. Take a net $(X_\alpha)\subset\cX$ such that $X_\alpha\xrightarrow{o}X$ for some $X\in\cX$. Then, we have
$X_\alpha^-\wedge nU\xrightarrow{o}X^-\wedge nU$. In particular, we find a net $(Y_\beta)\subset\cX_+$ with $Y_\beta\downarrow0$ in $\cX$ and such that for every
$\beta$ there exists $\alpha_0$ satisfying $\abs{X_\alpha^-\wedge nU-X^-\wedge nU}\leq Y_\beta$ for all $\alpha\geq\alpha_0$. Since we
clearly have $\abs{X_\alpha^-\wedge nU-X^-\wedge nU}\leq nU$ for every $\alpha$, it follows that
\[
\abs{X_\alpha^-\wedge nU-X^-\wedge nU}\leq Y_\beta\wedge nU
\]
for all $\alpha\geq\alpha_0$. We now infer from $Y_\beta\wedge nU\downarrow0$ in $\cI_U$ that $X_\alpha^-\wedge
nU\xrightarrow{o}X^-\wedge nU$ in $\cI_U$. As a result, the Fatou property of $\rho$ implies that
\[
\rho_n(X) = \rho(-X^-\wedge nU) \leq \liminf_\alpha\rho(-X_\alpha^-\wedge nU) = \liminf_\alpha\rho_n(X_\alpha).
\]
This shows that $\rho_n$ has the Fatou property. Since for every $m\in\R$ we have
\[
\{\overline{\rho}\leq m\} = \bigcap_{n\in\N}\{\rho_n\leq m\},
\]
it follows that $\overline{\rho}$ also has the Fatou property by Lemma \ref{lem: properties rho}. It remains to prove that
$\overline{\rho}$ is the unique surplus-invariant monotone extension of $\rho$ with the Fatou property. But this is a direct
consequence of Lemma~\ref{lem: unique surplus invariant extension}.

\smallskip

Finally, we show that the extension preserves (quasi)convexity. To this end, assume first that $\rho$ is quasiconvex. Pick any
$X,Y\in\cX$ and $\lambda\in[0,1]$. For every $n\in\N$ define $X_n=X^+\wedge nU-X^-\wedge nU$ and $Y_n=Y^+\wedge nU-Y^-\wedge nU$ and
observe that $X_n^-=X^-\wedge nU$ as well as $Y_n^-=Y^-\wedge nU$. Clearly, $X_n\xrightarrow{o}X$ and $Y_n\xrightarrow{o}Y$, so that
$\lambda X_n+(1-\lambda) Y_n\xrightarrow{o}\lambda X+(1-\lambda) Y$. Thus, by the Fatou property of $\overline{\rho}$, we have
\begin{align*}
&\overline{\rho}(\lambda X+(1-\lambda)Y)\\
\leq &\liminf_{n\to\infty} \overline{\rho}(\lambda X_n+(1-\lambda)Y_n)=\liminf_{n\to\infty} \rho(\lambda X_n+(1-\lambda)Y_n)\\
\leq & \liminf_{n\to\infty} \max(\rho(X_n),\rho(Y_n))=\liminf_{n\to\infty} \max(\rho(-X_n^-),\rho(-Y_n^-))\\
=&\liminf_{n\to\infty}\max(\rho(-X^-\wedge nU),\rho(-Y^-\wedge nU))=\max(\overline{\rho}(X),\overline{\rho}(Y)),
\end{align*}
where in the second inequality we have used the quasiconvexity of $\rho$. This proves that $\overline{\rho}$ is quasiconvex. The proof
of convexity can be obtained along the same lines.
\end{proof}

\smallskip

The next corollary provides a dual formulation of the above extension in the convex case. For this to be possible, we need to require
that $\cX^\sim_n$ is separating.

\begin{corollary}
\label{cor: convex extension}
Let $U$ be a weak order unit of $\cX$ and assume that $\cI$ is a separating ideal of $\cX^\sim_n$. Let $\rho:\cI_U\to\R\cup\{\infty\}$
be a convex surplus-invariant monotone functional with the Fatou property and let $\overline{\rho}:\cX\to\R\cup\{\infty\}$ be its
unique extension as determined in Theorem \ref{theo: general extension}. Then, for every $X\in\cX$ we have
\[
\overline{\rho}(X) = \sup_{\varphi\in\cI_-}\big(\varphi(-X^-)-\rho^\ast(\varphi)\big),
\]
where
\[
\rho^\ast(\varphi) = \sup_{X\in\cI_U}\big(\varphi(X)-\rho(X)\big).
\]
\end{corollary}
\begin{proof}
It follows from Remark \ref{rem: dual representation} that $\overline{\rho}$ has the desired representation with
\[
\rho^\ast(\varphi) = \sup_{X\in\cX}\big(\varphi(X)-\overline{\rho}(X)\big)
\]
for $\varphi\in\cI_-$. Thus, it remains to show that
\[
\sup_{X\in\cX}\big(\varphi(X)-\overline{\rho}(X)\big) = \sup_{X\in\cI_U}\big(\varphi(X)-\rho(X)\big)
\]
for every fixed $\varphi\in\cI_-$. Clearly, we only need to establish the inequality ``$\leq$''. To this effect, take $X\in\cX$ and
set $X_n=X^+\wedge nU-X^-\wedge nU$ for all $n\in\N$. Note that $X_n\in\cI_U$ for every $n\in\N$. Since $X_n\xrightarrow{o}X^+-X^-=X$,
we have $\varphi(X_n)\to\varphi(X)$ and
\[
\rho(X_n) = \rho(-X_n^-) = \rho(-X^-\wedge nU) \to \overline{\rho}(X).
\]
This implies that the inequality ``$\leq$'' holds and concludes the proof.
\end{proof}

\smallskip

%%%%%%%%%%%%%%%%%%%%%%%%%%%%%%%%%%%%%%%%%%%%%%

\subsection{Surplus invariance subject to positivity}

In this section we focus on a weaker notion of surplus invariance that is compatible with $S$-additivity. As such, this notion is
particularly appealing in a capital adequacy context. From now on we fix an element $S\in\cX_+\backslash\{0\}$.

\smallskip

We start by highlighting that, as remarked in Staum~(2013), there is a fundamental tension between surplus invariance and
$S$-additivity.

\begin{proposition}
There exists no monotone functional $\rho:\cX\to\R\cup\{\infty\}$ that is surplus invariant, $S$-additive, and such that
$\rho(X)\in\R$ for some $X\in\cX$.
\end{proposition}
\begin{proof}
If $\rho$ satisfies the above properties, then we must have
\[
\rho(X)-m = \rho(X+mS) = \rho(-(X+mS)^-) \geq \rho(0)
\]
for every $X\in\cX$ and $m\in\R$, which is only possible if $\rho(X)=\infty$ for all $X\in\cX$.
\end{proof}

\smallskip

The previous result shows that one has to make a choice between surplus invariance and $S$-additivity. As discussed above, the
property of $S$-additivity plays a crucial role in a capital adequacy context because of its operational interpretation. This explains
why surplus invariance, which was initially listed among the defining axioms of a coherent risk measure in Artzner et al.~(1996), was
eventually replaced by $S$-additivity in Artzner et al.~(1999). The key idea to reconcile surplus invariance and $S$-additivity in a
capital adequacy framework is to observe that a capital requirement rule should depend on the default profile of a financial position
only when the position is {\em not acceptable} from a regulatory perspective and thus a capital injection is actually required. This
leads to introducing the following weaker form of surplus invariance.

\begin{definition}
A functional $\rho:\cX\to\overline{\R}$ is {\em surplus invariant subject to positivity} if for every
$X\in\cX$ we have
\[
\rho(X)>0 \ \implies \ \rho(X)=\rho(-X^-).
\]
\end{definition}

\smallskip

The notion of surplus invariance subject to positivity was introduced, in a slightly different form, in Koch-Medina et al.~(2015) in
the setting of bounded univariate positions. By adapting Proposition 3.10 in that paper, one can show that, for $S$-additive monotone
functionals, surplus invariance subject to positivity corresponds to the surplus invariance of the underlying ``acceptance set''.

\begin{proposition}
\label{prop: characterization si subject to pos}
For an $S$-additive monotone functional $\rho:\cX\to\R\cup\{\infty\}$ the following are equivalent:
\begin{enumerate}
  \item[(a)] $\rho$ is surplus invariant subject to positivity.
  \item[(b)] $\{\rho\leq0\}$ is surplus invariant.
\end{enumerate}
\end{proposition}
\begin{proof}
Assume that {\em (a)} holds and take $X\in\{\rho\leq0\}$. Note that $\rho(X)\in\R$. Then, for every $\varepsilon>0$ we have
\[
\rho(X+(\rho(X)-\varepsilon)S) = \rho(X)-(\rho(X)-\varepsilon) = \varepsilon > 0
\]
by $S$-additivity. It follows from our assumption that
\[
\rho(-(X+(\rho(X)-\varepsilon)S)^-) = \rho(X+(\rho(X)-\varepsilon)S) = \varepsilon.
\]
Since $-X^-\geq-(X+(\rho(X)-\varepsilon)S)^-$, we infer from monotonicity that
\[
\rho(-X^-) \leq \rho(-(X+(\rho(X)-\varepsilon)S)^-) = \varepsilon.
\]
This holds for every $\varepsilon>0$ and yields $\rho(-X^-)\leq0$, proving that {\em (b)} is satisfied.

\smallskip

Conversely, assume that {\em (b)} holds and take an arbitrary $X\in\cX$ with
$\rho(X)>0$. Note that $\rho(X)\leq\rho(-X^-)$ by monotonicity. If $\rho(X)=\infty$, then
\[
\rho(X) = \rho(-X^-) = \infty.
\]
Hence, assume that $\rho(X)<\infty$ and take an arbitrary $m>\rho(X)$. Since $\rho(X+mS)=\rho(X)-m<0$, it follows from surplus
invariance that $-(X+mS)^-\in\{\rho\leq0\}$. Then, the monotonicity of $\{\rho\leq0\}$ and the fact that $-X^-+mS\geq-(X+mS)^-$ imply
that $-X^-+mS\in\{\rho\leq0\}$, which in turn yields $\rho(-X^-)\leq m$ by virtue of $S$-additivity. Since $m$ was arbitrary, we infer
that
\[
\rho(X) \geq \rho(-X^-).
\]
As said above, the converse inequality holds by monotonicity and we therefore conclude that $\rho(X)=\rho(-X^-)$, proving
that {\em (a)} holds.
\end{proof}

\smallskip

Since each sublevel set of an $S$-additive risk measure is simply a translate of the zero sublevel set, see Lemma~\ref{lem: properties
rho}, we can apply Theorem~\ref{closed-top} to immediately derive the following characterization of the Fatou property under surplus
invariance subject to positivity. Note that a quasiconvex $S$-additive functional is automatically convex. This is why, differently
from Theorem~\ref{s-i functionals} above, we state the next theorem for convex functionals.

\begin{theorem}
\label{theo: dual representation S additive}
For a convex $S$-additive monotone functional $\rho:\cX\to\R\cup\{\infty\}$ that is
surplus invariant subject to positivity the following statements are equivalent:
\begin{enumerate}
  \item[(a)] $\rho$ has the Fatou property.
  \item[(b)] $\rho$ has the super Fatou property.
\end{enumerate}
If $\cX^\sim_n$ is separating, then the above are also equivalent to:
\begin{enumerate}
  \item[(c)] $\rho$ is $\sigma(\cX,\cI)$ lower semicontinuous for every separating ideal $\cI\subset\cX^\sim_n$.
  \item[(d)] $\rho$ is $\sigma(\cX,\cI)$ lower semicontinuous for some separating ideal $\cI\subset\cX^\sim_n$.
\end{enumerate}
\end{theorem}

\smallskip

\begin{remark}
\label{rem: dual representation S additive}
A variant of Remark \ref{rem: dual representation} applies to the above theorem. In this case, the condition
$\rho^\ast(\varphi)<\infty$ is easily seen to imply $\varphi(S)=-1$, so that the dual domain in the representation \eqref{eq: dual
representation} can be restricted to those functionals in $\cI$ that are negative and satisfy $\varphi(S)=-1$.
\end{remark}

\smallskip

The preceding extension results can still be obtained under the weaker property of surplus invariance subject to positivity in the
context of $S$-additive functionals.

\begin{lemma}
\label{lem: unique extension si subject to pos}
Let $U$ be a weak order unit of $\cX$ and consider two monotone $S$-additive functionals $\rho_1,\rho_2:\cX\to\R\cup\{\infty\}$ that
are surplus invariant subject to positivity and have the Fatou property. Then, $\rho_1=\rho_2$ on $\cI_U$ implies $\rho_1=\rho_2$ on
$\cX$.
\end{lemma}
\begin{proof}
Fix $X\in\cX$ and choose $m>0$ large enough so that
\[
\rho_1(X-mS)=\rho_1(X)+m>0 \ \ \ \mbox{and} \ \ \ \rho_2(X-mS)=\rho_2(X)+m>0.
\]
Then, a similar argument as in Lemma~\ref{lem: unique surplus invariant extension} shows that $\rho_1(X-mS)=\rho_2(X-mS)$, which
yields $\rho_1(X)=\rho_2(X)$ by $S$-additivity.
\end{proof}

\smallskip

For every set $\cA\subset\cX$ we denote by $\cl_o(\cA)$ the {\em order closure} of $\cA$, i.e.\ the set of all limits of order-convergent nets consisting of elements of $\cA$.

\begin{theorem}
\label{theo: extension si subject pos}
Let $U$ be a weak order unit of $\cX$ and assume $S\in(\cI_U)_+\backslash\{0\}$. Then, every monotone $S$-additive functional
$\rho:\cI_U\to\R\cup\{\infty\}$ that is surplus invariant subject to positivity and has the Fatou property can be uniquely extended to
a monotone $S$-additive functional $\overline{\rho}:\cX\to\R\cup\{\infty\}$ that is surplus invariant subject to positivity and has
the Fatou property. In particular, for every $X\in\cX$ we have
\[
\overline{\rho}(X) = \inf\{m\in\R \,; \ -(X+mS)^-\in\cl_o(\cA_-)\},
\]
where $\cA=\{\rho\leq 0\}$. In addition, if $\rho$ is convex, then $\overline{\rho}$ is convex as well.
\end{theorem}
\begin{proof}
Set $\cA=\{\rho\leq 0\}$. Being monotone and surplus invariant as well as order closed, it follows from Proposition \ref{s-i} and
Proposition \ref{prop: closedness} that $\cA$ can be decomposed as $\cA=(\cI_U)_+-\cD$, where $\cD=-\cA_-$ is order closed and solid
in $(\cI_U)_+$. Now, define
\[
\cD' = \{X\in\cX \,; \ \mbox{there exists a net $(X_\alpha)\subset\cD$ such that $X_\alpha\xrightarrow{o}X$ in $\cX$}\}.
\]
Since $\cD\subset\cX_+$, it is clear that $\cD'\subset\cX_+$. We claim that $\cD'$ is solid in $\cX_+$. To see this, take $X\in\cD'$
and $Y\in\cX_+$ such that $Y\leq X$. Note that, being solid in $(\cI_U)_+$, the set $\cD$ is solid in $\cX_+$ because $\cI_U$ is an
ideal of $\cX$. Thus, by Lemma \ref{con}, there exists a net $(X_\alpha)\subset\cD$ such that $X_\alpha\uparrow X$ in $\cX$, so that
$Y\wedge X_\alpha\uparrow Y\wedge X=Y$ in $\cX$. Since $Y\wedge X_\alpha\in\cD$ for all $\alpha$ by solidity of $\cD$ in $\cX_+$, we
conclude that $Y\in\cD'$.

\smallskip

Next, we show that $\cD'$ is order closed. Let $(X_\alpha)\subset\cD'$ and $X\in\cX$ be such that $X_\alpha\xrightarrow{o}X$. In
view of Lemma \ref{con} and of the solidity of $\cD'$ in $\cX_+$, we may assume that $X_\alpha\uparrow X$. Another application of
Lemma \ref{con} shows that, for each $\alpha$, there exists a subset $\cZ_\alpha\subset\cD$ that is directed upward and satisfies
$\sup\cZ_\alpha=X_\alpha$ in $\cX$. Then, the set $\bigcup_\alpha \cZ_\alpha\subset \cD$ is also directed upward (and thus can be
treated as an increasing net indexed over itself) and $\sup\bigcup_\alpha\cZ_\alpha=\sup_\alpha\sup \cZ_\alpha=\sup_\alpha X_\alpha=X$
in $\cX$. It follows that $X\in\cD'$ by definition of $\cD'$. This establishes order closedness.

\smallskip

We claim that $\cI_U\cap\cD'=\cD$. It is clear that $\cD\subset\cI_U\cap\cD'$. Conversely, take any $X\in\cI_U\cap\cD'$. Then, there
exists a net $(X_\alpha)\subset \cD$ such that $X_\alpha\uparrow X$ in $\cX$. Since $X\in\cI_U$ and since
$(X_\alpha)\subset\cD\subset\cI_U$ and $\cI_U$ is an ideal of $\cX$, it is easily verified that $X_\alpha\uparrow X$ in $\cI_U$.
Therefore, $X\in \cD$ because $\cD$ is order closed in $\cI_U$.

\smallskip

Now, set $\cA'=\cX_+-\cD'$ and observe that $\cA'$ is monotone, surplus invariant, and order closed by Proposition \ref{s-i} and
Proposition \ref{prop: closedness}. Then, for every $X\in\cX$ define
\begin{equation}
\label{eq: extension S additive case}
\overline{\rho}(X) = \inf\{m\in\R \,; \ X+mS\in\cA'\}.
\end{equation}
First of all, we show that $\overline{\rho}(X)>-\infty $ for every $X\in\cX$. To the contrary, suppose that
$\overline{\rho}(X)=-\infty$ for some $X\in\cX$ so that $X-nS\in\cA'$ for all $n\in\N$ by the monotonicity of $\cA'$. Note that
$X^+-nS\in\cA'$ and thus $(nS-X^+)^+=(X^+-nS)^-\in\cD'$ for all $n\in\N$. Now, fix $n_0\in\N$. Then, for every choice of $n\geq n_0$
we easily see that $n_0(S-\frac{X^+}{n})^+\leq n(S-\frac{X^+}{n})^+=(nS-X^+)^+\in\cD'$ as well as $n_0(S-\frac{X^+}{n})^+\leq n_0S\in
\cI_U$. Since $\cD'\cap \cI_U=\cD$, as established above, we infer from solidity that $n_0(S-\frac{X^+}{n})^+\in\cD$. Since
$n_0(S-\frac{X^+}{n})^+\xrightarrow{o}n_0S$ in $\cI_U$ and $\cD$ is order closed in $\cI_U$, it follows that
$n_0S\in\cD$ and, hence, $-n_0S\in\cA$. But this would hold for every $n_0\in\N$, which would lead to the contradiction
$\rho(0)=-\infty$. As a result, we must have $\overline{\rho}(X)>-\infty $ for every $X\in\cX$.

\smallskip

We claim that $\overline{\rho}$ has all the desired properties. To show that $\overline{\rho}$ extends $\rho$, take an arbitrary
$X\in\cI_U$ and note that for every $m\in\R$ we have $X+mS\in\cI_U$ and thus $(X+mS)^-\in\cI_U$. This implies that $(X+mS)^-\in\cD'$
is equivalent to $(X+mS)^-\in \cD$. Since $(X+mS)^-\in\cD'$ is equivalent to $X+m S\in\cA'$ and $(X+mS)^-\in \cD$ is equivalent to
$X+mS\in \cA$, it follows from~\eqref{eq: S additivity} that $\overline{\rho}(X)=\rho(X)$.

\smallskip

It is clear that $\overline{\rho}$ is monotone and $S$-additive. Since $\cA'$ is surplus invariant and order closed, it follows from
Lemma~\ref{lem: properties rho} and Proposition~\ref{prop: characterization si subject to pos} that, to prove that $\overline{\rho}$
is surplus invariant subject to positivity and has the Fatou property, it suffices to show that $\cA'=\{\overline{\rho}\leq 0\}$. The
inclusion ``$\subset$'' is clear. To show the converse inclusion ``$\supset$'', take an arbitrary $X\in\{\overline{\rho}\leq 0\}$ and
note that $X+m_nS\in\cA'$ for a decreasing sequence $(m_n)\subset\R$ converging to $\overline{\rho}(X)$. Since $X+m_nS\downarrow
X+\overline{\rho}(X)S$, we infer that $X+\overline{\rho}(X)S\in\cA'$ by order closedness. Since $\overline{\rho}(X)\leq0$, the
monotonicity of $\cA'$ implies that $X\in\cA'$ and this yields the desired inclusion.

\smallskip

The uniqueness of the extension follows from Lemma~\ref{lem: unique extension si subject to pos}. Moreover, the above representation
of $\overline{\rho}$ follows from \eqref{eq: extension S additive case} and from Proposition \ref{s-i}. Finally, note that, if $\rho$
is convex, then $\cA$ is convex and so is $\cD$ by Corollary~\ref{cor: convexity}. As a result, one sees easily that $\cD'$, and
therefore $\cA'$, is also convex and this implies that $\overline{\rho}$ is convex as well.
\end{proof}

\smallskip

\begin{corollary}
Let $U$ be a weak order unit of $\cX$ and assume $S\in(\cI_U)_+\backslash\{0\}$. Moreover, assume that $\cI$ is a separating ideal of $\cX^\sim_n$. Let $\rho:\cI_U\to\R\cup\{\infty\}$ be a convex monotone $S$-additive functional that is surplus-invariant subject to positivity and has the Fatou property and let $\overline{\rho}:\cX\to\R\cup\{\infty\}$ be its unique extension as determined in Theorem \ref{theo: extension si subject pos}. Then, for every $X\in\cX$ we have
\[
\overline{\rho}(X) = \sup_{\varphi\in\cI_-,\,\varphi(S)=-1}\big(\varphi(X)-\rho^\ast(\varphi)\big),
\]
where
\[
\rho^\ast(\varphi) = \sup_{X\in\cI_U}\big(\varphi(X)-\rho(X)\big).
\]
\end{corollary}
\begin{proof}
It follows from Theorem \ref{theo: dual representation S additive} and Remark \ref{rem: dual representation S additive} that
$\overline{\rho}$ has the desired representation with
\[
\rho^\ast(\varphi) = \sup_{X\in\cX}\big(\varphi(X)-\overline{\rho}(X)\big)
\]
for every $\varphi\in\cI_-$ such that $\varphi(S)=-1$. It remains to show that for every $\varphi\in\cI_-$ such that $\varphi(S)=-1$
we have
\[
\sup_{X\in\cX}\big(\varphi(X)-\overline{\rho}(X)\big) = \sup_{X\in\cI_U}\big(\varphi(X)-\rho(X)\big).
\]
Clearly, we only need to prove the inequality ``$\leq$''. To this effect, take arbitrary $X\in\cX$ and $m<\overline{\rho}(X)$ and set
$Y=-(X+mS)^-$ and $Y_n=-(X+mS)^-\wedge nU$ for $n\in\N$. Observe that $Y_n\in\cI_U$ for every $n\in\N$. Since $Y_n\downarrow Y$, we have $\varphi(Y_n)\to\varphi(Y)$. In addition, the Fatou property and monotonicity imply that
\[
\overline{\rho}(Y) \leq \liminf_{n\to\infty}\overline{\rho}(Y_n) \leq \limsup_{n\to\infty}\overline{\rho}(Y_n) \leq \overline{\rho}(Y),
\]
so that $\rho(Y_n)=\overline{\rho}(Y_n)\to\overline{\rho}(Y)$. Now, since $\overline{\rho}(X+mS)=\overline{\rho}(X)-m>0$, we have
$\overline{\rho}(X+mS)=\overline{\rho}(Y)$ and thus
\[
\varphi(X)-\overline{\rho}(X) = \varphi(X+mS)-\overline{\rho}(X+mS) \leq \varphi(Y)-\overline{\rho}(Y) =
\lim_{n\to\infty}\varphi(Y_n)-\rho(Y_n),
\]
where we used that $X+mS\geq Y$ and that $\varphi$ is negative and satisfies $\varphi(S)=-1$. The desired inequality follows
immediately.
\end{proof}

%%%%%%%%%%%%%%%%%%%%%%%%%%%%%%%%%%%%%%%%

\section{Applications}

In this section we apply our general results to a variety of concrete model spaces used in mathematical finance. Throughout this section we fix a measurable space $(\Omega,\cF)$ and equip $\R$ with its canonical Borel structure. The set of measurable functions $X:\Omega\to\R$ is denoted by $\cL^0$. The elements of $\cL^0$ represent financial positions (e.g.\ profit-and-losses, equity values, payoffs of financial assets or portfolios, monetary exposures) expressed in a given currency at a future pre-specified point in time. The subset of $\cL^0$ consisting of bounded measurable functions is denoted by $\cL^\infty$.

\medskip

{\bf Models with a fixed underlying probability}. Let $\probp$ be a probability measure on $(\Omega,\cF)$. The vector space of all
equivalence classes of functions in $\cL^0$ with respect to almost-sure equality (with respect to $\probp$) is denoted by $L^0$. As
usual, we will not distinguish between equivalence classes and their representative elements. The set $L^0$ is an order-complete vector lattice with the countable sup property under the partial order
\[
X\geq Y \ \iff \ \probp(X\geq Y)=1.
\]
Moreover, it becomes a complete metric space when equipped with the topology of convergence in probability. Let $\cX$ be an ideal of $L^0$. Each functional $\varphi\in\cX_n^\sim$ can be represented by some $Y\in L^0$ via
\[
\varphi(X)=\E_\probp[XY],
\]
where $Y$ satisfies $\E_\probp[\abs{XY}]<\infty$ for all $X\in\cX$ and is uniquely determined on the support of $\cX$. The converse also holds. It follows that $\cX_n^\sim$ can be identified with an ideal of $L^0$ and, thus, has the countable sup property. Note that for
every sequence $(X_n)\subset\cX$ and every $X\in\cX$ we have
\[
X_n\xrightarrow{o}X \ \iff \ \mbox{$\abs{X_n}\leq Y$ for some $Y\in\cX$ and
$X_n\xrightarrow{a.s.}X$}.
\]
In other words, order convergence corresponds to dominated almost-sure convergence. Similarly, we have
\[
X_n\xrightarrow{uo}X \ \iff \ X_n\xrightarrow{a.s.}X,
\]
showing that unbounded-order convergence corresponds to almost-sure convergence.

\smallskip

A function $\Phi:[0,\infty)\rightarrow[0,\infty]$ is called an Orlicz
function if it is convex, increasing, left-continuous, and satisfies $\Phi(0)=0$. The Orlicz space $L^\Phi$ is the ideal of $L^0$
consisting of all random variables
$X\in L^0$ such that
\[
\norm{X}_\Phi = \inf\{\lambda>0 \,; \ \E_\probp[\Phi(|X|/\lambda)]\leq 1\} < \infty.
\]
The function $\|\cdot\|_{\Phi}$ is a norm, called the Luxemburg norm. If $\Phi(x)=x^p$ for some $p\in[1,\infty)$, then we have
$L^\Phi=L^p$. If $\Phi(x)=0$ for $0\leq x\leq 1$ and $\Phi(x)=\infty$ for $x>1$, then we have $L^\Phi=L^\infty$. The Orlicz heart of
$L^\Phi$ is the ideal defined by
\[
H^\Phi = \{X\in L^\Phi \,; \ \mbox{$\E_\probp[\Phi(|X|/\lambda)]<\infty$ for every $\lambda>0$}\}.
\]
If $\Phi$ satisfies the so-called $\Delta_2$-condition, i.e.~there exist $x_0>0$ and $k\in\R$ such that $\Phi(2x)<k\Phi(x)$ for every
$x\geq x_0$, then $L^\Phi=H^\Phi$. The conjugate function of $\Phi$ is the Orlicz function $\Phi^\ast:[0,\infty)\rightarrow[0,\infty]$
given by
\[
\Phi^\ast(y) = \sup\{xy-\Phi(x) \,; \ x\geq 0\}.
\]
We have the standard identification $(L^\Phi)^\sim_n=L^{\Phi^\ast}$. In particular, we have $(L^p)^\sim_n=L^q$ where $1\leq p,q\leq
\infty$ and $\frac{1}{p}+\frac{1}{q}=1$. We refer to Edgar and Sucheston (1992) for more details about Orlicz spaces.

\medskip

{\bf Models with multiple underlying probabilities}. Consider a family $\cP$ of probabilities over $(\Omega,\cF)$ and the associated capacity
$c:\cF\to[0,1]$ defined by
\[
c(E) = \sup_{\probp\in\cP}\probp(E).
\]
We denote by $L^0_c$ the vector space of all equivalence classes of functions in $\cL^0$ with respect to quasi-sure equality (with
respect to $\cP$) defined by
\[
X\sim Y \ \iff \ \mbox{$\probp(X=Y)=1$ for every $\probp\in\cP$}.
\]
Once again, we do not distinguish between equivalence classes and their representative elements. The set $L^0_c$ is a vector lattice under the partial order defined by
\[
X\geq Y \ \iff \ \mbox{$\probp(X\geq Y)=1$ for every $\probp\in\cP$}.
\]
For $1\leq p\leq\infty$ we denote by $L^p_c$ the ideal of $L_c^0$ consisting of all $X\in L^0_c$ such that
\[
\|X\|_p=\sup_{\probp\in\cP}\|{X}\|_{L^p(\probp)}<\infty.
\]
The space $L^p_c$ is a Banach lattice with norm $\|\cdot\|_p$. It is known that many fundamental properties of classical $L^p$ spaces, most notably the countable sup property and order completeness, may fail in the robust setting as soon as the family $\cP$ is nondominated. In particular, the lack of the countable sup property makes it difficult to ensure the existence of a supremum for an arbitrary subset of $L^p_c$ that is bounded from above, unless the subset is countable (in which case the supremum in $L^p_c$ is simply given by the pointwise supremum). The following lemma provides a characterization of the order completeness of $L^0_c$ in terms of the order completeness of $L^p_c$ spaces.

\begin{lemma}
\label{lem: characterization order completeness robust}
The following statements are equivalent:
\begin{enumerate}
\item[(a)] $L^0_c$ is order complete.
\item[(b)] Every ideal of $L^0_c$ is order complete.
\item[(c)] $L^p_c$ is order complete for some $1\leq p<\infty$.
\item[(d)] $L^\infty_c$ is order complete.
\end{enumerate}
\end{lemma}
\begin{proof}
It is clear that {\em (a)} implies {\em (b)} and that {\em (b)} implies {\em (c)}. Since $L^\infty_c$ is an ideal of $L^p_c$, we also have that {\em (c)} implies {\em (d)}. To conclude the proof, assume that $L^\infty_c$ is order complete and let $\cA$ be a subset of $(L_c^0)_+$ that is bounded above by some
$Y\in (L_c^0)_+$. For every $n\in\N$ the set $\cA_n=\{X\wedge n\one_\Omega \,; \ X\in\cA\}$ is contained in $(L_c^\infty)_+$ and is bounded above
by $Y\wedge n\one_\Omega\in L_c^\infty$. For each $n\in\N$ we denote by $X_n$ the supremum of $\cA_n$ in $L_c^\infty$. Being bounded above by
$Y$, the set $(X_n)$ admits a supremum in $L_c^0$, denoted by $X_0$, due to the remark preceding the lemma. It is now straightforward
to verify that $X_0$ is the supremum of $\cA$ in $L_c^0$.
\end{proof}

\smallskip

A simple tractable characterization of the order-continuous dual $(L^p_c)_n^\sim$ is typically not available in a robust setting. This is because, for $1\leq p,q\leq\infty$ such that $\frac{1}{p}+\frac{1}{q}=1$, the identification $L^q_c=(L^p_c)_n^\sim$ no longer holds in general. For example, if $\cP$ consists of all the Dirac measures associated to the elements of $\Omega$, then one sees that $L^0_c=\cL^0$ and $L^p_c=\cL^\infty$ for all $1\leq p\leq\infty$. In particular, the lack of the countable sup property forces us to deal with order-convergent nets instead of sequences when it comes to determining $(L^p_c)_n^\sim$. In what follows, we avoid this problem by focusing on a special tractable ideal of $(L^p_c)_n^\sim$.

\smallskip

Let $\mathrm{ca}_c$ be the space of all countably-additive signed measures $\mu$ over $(\Omega,\cF)$ such that $\mu(E)=0$ whenever
$E\in\cF$ and $c(E)=0$. The space $\mathrm{ca}_c$ is a Banach lattice with respect to the total variation norm and the setwise partial
order. Given $\probp\in\cP$ and $Z\in\cL^\infty$, we define a signed measure $\mu_{\mathbb{P},Z}\in\mathrm{ca}_c$ by setting
\[
\mu_{\mathbb{P},Z}(E)=\mathbb{E}_\mathbb{P}[\one_EZ].
\]
We denote by $\mathrm{ca}_c^\infty$ the linear subspace of $\mathrm{ca}_c$ generated by such measures, i.e.
\[
\mathrm{ca}_c^\infty=\mathrm{span}(\{\mu_{\mathbb{P},Z} \,; \ \mathbb{P}\in\mathcal{P}, \ Z\in\cL^\infty\}).
\]

\begin{lemma}
\label{lem: ca infty ideal}
The set $\mathrm{ca}_c^\infty$ is an ideal of $\mathrm{ca}_c$.
\end{lemma}
\begin{proof}
It suffices to show that $\mathrm{ca}_c^\infty$ is solid. To this end, let
$\mu\in\mathrm{ca}_c^\infty$ and assume that $\mu=\sum_{i=1}^n a_i\mu_{\probp_i,Z_i}$ for $\probp_1,\dots,\probp_n\in\cP$ and
$Z_1,\dots,Z_n\in\cL^\infty$. Moreover, take $\nu\in\mathrm{ca}_c$ and assume that $\abs{\nu}\leq\abs{\mu}$. This implies that
\[
\abs{\nu} \leq \sum_{i=1}^n\abs{a_i}\mu_{\probp_i,\abs{Z_i}}.
\]
Hence, by Radon-Nikodym, we find $Z\in\cL^\infty$ such that $\abs{Z}\leq 1$ and
\[
\nu = \sum_{i=1}^n\abs{a_i}\mu_{\probp_i,Z\abs{Z_i}}.
\]
This proves that $\nu\in\mathrm{ca}_c^\infty$ and establishes solidity.
\end{proof}

\smallskip

Now, consider the natural duality embedding of $\mathrm{ca}_c$ into the norm dual of $L_c^\infty$, respectively of
$\mathrm{ca}_c^\infty$ into the norm dual of $L^p_c$ for $1\leq p\leq\infty$, given by
\[
\langle X,\mu\rangle = \E_\mu[X].
\]
The duality is well defined in the sense that $\E_\mu[X]=\E_\mu[Y]$ whenever $X=Y$ quasi surely. These embeddings also respect the order structure in the sense that the supremum of two signed measures in $\mathrm{ca}_c$, respectively in $\mathrm{ca}_c^\infty$, coincides with their supremum as functionals in the above norm duals.

\smallskip

In our applications we follow Maggis et al.~(2018) and work under the {\bf assumption}
\[
L_c^\infty=(\mathrm{ca}_c)^\ast,
\]
where $(\mathrm{ca}_c)^\ast$ denotes the norm dual of $\mathrm{ca}_c$. As dual spaces are order complete, this assumption implies that
$L_c^\infty$ is order complete. An application of Theorem 4.60 in Aliprantis and Burkinshaw (2006) shows that this assumption also
implies that
\[
(L_c^\infty)_n^\sim=\mathrm{ca}_c.
\]
By Theorem 2.4.22 in Meyer-Nieburg (1991), one sees that the converse is also true, namely, if $L_c^\infty$ is order complete and $(L_c^\infty)_n^\sim=\mathrm{ca}_c$, then $L_c^\infty=\big((L_c^\infty)_n^\sim\big)_n^\sim=(\mathrm{ca}_c)_n^\sim=(\mathrm{ca}_c)^\ast$. However, Proposition 3.10 in Maggis et al.~(2018) surprisingly proves that order completeness of $L_c^\infty$ alone is sufficient to yield $L_c^\infty=(\mathrm{ca}_c)^\ast$.

\smallskip

%%%%%%%%%%%%%%%%%%%%%%%%%%%%%%%%%%%%%%%%%%%%%%%%

\subsection{Dual characterizations of the Fatou property under surplus invariance}

In this section we provide a variety of dual characterizations of the Fatou property based on Theorems \ref{s-i functionals} and \ref{theo: dual representation S additive}. To highlight the embedding of our results in the literature we start by focusing on model spaces with a dominating probability. The prototype dual characterization in this setting was obtained by Delbaen~(2002) in the context of the space $\cX=L^\infty$. In this case, it was shown that for every convex monotone functional the Fatou property is equivalent to $\sigma(L^\infty,L^1)$ lower semicontinuity. As such, the Fatou property allows for a ``nice'' dual representation where the dual elements can be identified with countably-additive measures on $(\Omega,\cF)$. Our first result shows that, under surplus invariance (subject to positivity), we can always restrict the dual space to the smaller domain $L^\infty$. This entails a sharper dual representation in terms of countably-additive measures that have bounded Radon-Nikodym derivatives with respect to $\probp$.

\begin{theorem}
\label{theo: fatou in L infty}
If $\rho:L^\infty\to\R\cup\{\infty\}$ is {\bf either} quasiconvex, surplus invariant, and monotone {\bf or} convex, surplus invariant
subject to positivity, monotone, and $S$-additive with $S\in L^\infty_+\backslash\{0\}$, then the following statements are
equivalent:
\begin{enumerate}
  \item[(a)] $\rho$ has the Fatou property.
  \item[(b)] $\rho$ has the super Fatou property.
  \item[(c)] $\rho$ is $\sigma(L^\infty,L^1)$ lower semicontinuous.
  \item[(d)] $\rho$ is $\sigma(L^\infty,L^\infty)$ lower semicontinuous.
\end{enumerate}
\end{theorem}

\smallskip

\begin{remark}
The general equivalence between {\em (a)} and {\em (c)} was used in Cont et al.~(2013) to obtain a dual representation of convex loss-based risk measures with the Fatou property in terms of dual elements in $L^1$. Since loss-based risk measures are surplus invariant, the above result allows us to refine that representation by restricting the space of dual elements to $L^\infty$. Moreover, it shows that the Fatou property is equivalent to the stronger super Fatou property, i.e.\ lower semicontinuity with respect to almost-sure convergence.
\end{remark}

\smallskip

Note that $\sigma(L^\infty,L^1)=\sigma(\cX,\cX^\sim_n)$ in the case that $\cX=L^\infty$ as above. As pointed out in Gao and Xanthos
(2018) and later in Gao et al.~(2016), the equivalence between the Fatou property and $\sigma(\cX,\cX^\sim_n)$ lower semicontinuity
breaks down for a quasiconvex functional (even for a standard convex cash-additive risk measure) in a general Orlicz space
$\cX=L^\Phi$, in which case $\cX^\sim_n=L^{\Phi^\ast}$. More specifically, Gao et al.~(2016) show that the equivalence breaks down
precisely when both $\Phi$ and $\Phi^\ast$ fail to satisfy the $\Delta_2$-condition; see also Delbaen and Owari (2016). As recently
established in Gao et al.~(2018), the equivalence can be nevertheless ensured in {\em every} Orlicz space under the additional
property of law invariance. In that case, the Fatou property is also equivalent to lower semicontinuity with respect to the (generally
much) coarser topologies $\sigma(L^\Phi,H^{\Phi^\ast})$ and $\sigma(L^\Phi,L^\infty)$. Our next result shows that the same conclusion
holds if we replace law invariance by surplus invariance (subject to positivity).

\begin{theorem}
\label{theo: fatou orlicz spaces}
If $\rho:L^\Phi\to\R\cup\{\infty\}$ is {\bf either} quasiconvex, surplus invariant, and monotone {\bf or} convex, surplus invariant
subject to positivity, monotone, and $S$-additive with $S\in L^\Phi_+\backslash\{0\}$, then the following statements are equivalent:
\begin{enumerate}
  \item[(a)] $\rho$ has the Fatou property.
  \item[(b)] $\rho$ has the super Fatou property.
  \item[(c)] $\rho$ is $\sigma(L^\Phi,L^{\Phi^\ast})$ lower semicontinuous.
  \item[(d)] $\rho$ is $\sigma(L^\Phi,H^{\Phi^\ast})$ lower semicontinuous.
  \item[(e)] $\rho$ is $\sigma(L^\Phi,L^\infty)$ lower semicontinuous.
\end{enumerate}
\end{theorem}

\smallskip

We turn to robust model spaces. In Maggis et al.~(2018), the authors studied dual representations of risk functionals on $L^\infty_c$ via the dual space $\mathrm{ca}_c$. Here, we show that surplus invariance (subject to positivity) allows us to establish dual representations for risk functionals on $L^p_c$ via $\mathrm{ca}_c^\infty$. This is appealing because $\mathrm{ca}_c^\infty$ consists of ``nice'' elements in $\mathrm{ca}_c$ that are directly associated with $\cP$ and that have bounded Radon-Nikodym derivatives. To this effect, we first need to show that $\mathrm{ca}_c^\infty$ is a separating ideal of $(L_c^p)_n^\sim$.

\begin{lemma}
The set $\mathrm{ca}_c^\infty$ is a separating ideal of $(L_c^p)_n^\sim$ for every $1\leq p\leq\infty$.
\end{lemma}
\begin{proof}
We know from Lemma \ref{lem: ca infty ideal} that $\mathrm{ca}_c^\infty$ is an ideal of $\mathrm{ca}_c=(L_c^\infty)_n^\sim$. Hence, assume that $1\leq p<\infty$. We first show that $\mathrm{ca}_c^\infty\subset(L^p_c)_n^\sim$ under the aforementioned duality. To this effect, take an arbitrary $\mu\in\mathrm{ca}_c^\infty$ and assume without loss of generality that $\mu\geq 0$. Let $(X_\alpha)\subset(L^p_c)_+$ and
$X\in(L^p_c)_+$ be such that $X_\alpha\uparrow X$ in $L^p_c$. It suffices to show that $\langle X_\alpha,\mu\rangle\uparrow \langle
X,\mu\rangle$. Suppose otherwise that there exists some $\varepsilon>0$ such that $\langle X_\alpha,\mu\rangle\leq \langle
X,\mu\rangle-\varepsilon$ for all $\alpha$. In this case, $\langle X_\alpha\wedge n,\mu\rangle\leq \langle X,\mu\rangle-\varepsilon$
for all $n\in\N$ and $\alpha$. Since $X_\alpha\wedge n\uparrow X\wedge n$ in $L^p_c$, hence in $L^\infty_c$, and since
$\mathrm{ca}_c^\infty\subset(L^\infty_c)_n^\sim$ under our assumption, we see that $\langle X_\alpha\wedge n,\mu\rangle\uparrow \langle X\wedge
n,\mu\rangle$. This implies that $\langle X\wedge n,\mu\rangle\leq \langle X,\mu\rangle-\varepsilon$ for all $n\in\N$. Letting
$n\to\infty$, we get a contradiction. This proves that $\mathrm{ca}_c^\infty\subset(L^p_c)_n^\sim$.

\smallskip

We claim that $\mathrm{ca}_c^\infty$ is an ideal of $(L_c^p)_n^\sim$. To show this, take $\varphi\in(L_c^p)_n^\sim$ and
$\mu\in\mathrm{ca}_c^\infty$ such that $|\varphi|\leq|\mu|$. By considering $\varphi^+$ and $\varphi^-$ and by replacing $\mu$ with
$|\mu|$, we may assume that $\varphi\geq0$ as well as $\mu\geq 0$, so that $0\leq\varphi\leq\mu$. For every $E\in\cF$ define
\[
\nu(E)=\varphi(\one_E).
\]
Take a pairwise disjoint sequence $(E_n)\subset\cF$. Since $\one_{\cup_{k=1}^nE_k}\uparrow \one_{\cup_{k=1}^\infty E_k}$, we have
\[
\sum_{k=1}^n\nu(E_k) = \sum_{k=1}^n\varphi(\one_{E_k}) =
\varphi\left(\sum_{k=1}^n\one_{E_k}\right)=\varphi(\one_{\cup_{k=1}^nE_k})\rightarrow \varphi(\one_{\cup_{k=1}^\infty E_k}) =
\nu(\cup_{k=1}^\infty E_k).
\]
In other words, $\nu$ is a countably-additive finite measure. Clearly, $\nu(E)=0$ whenever $c(E)=0$, so that $\nu\in\mathrm{ca}_c$.
Now, for every $X\in(L^p_c)_+$ we can find a sequence of simple functions $(X_n)\subset(L^p_c)_+$ such that $X_n\uparrow X$, so that
\[
\varphi(X) = \lim_{n\to\infty}\varphi(X_n) = \lim_{n\to\infty}\E_\nu[X_n] = \E_\nu[X].
\]
This yields $\varphi(X)=\E_\nu[X]$ for all $X\in L_c^p$ and shows that $\nu$ is identified with $\varphi$ in our duality. Clearly, $0\leq \nu\leq \mu$, implying that $\nu\in\mathrm{ca}_c^\infty$ since $\mathrm{ca}_c^\infty$ is an ideal of $\mathrm{ca}_c$.

\smallskip

To show that $\mathrm{ca}_c^\infty$ separates the points of $L^p_c$ for every $1\leq p\leq\infty$, it suffices to note that for every nonzero $X\in L_c^p$ we have $\probp(X\neq 0)>0$ for some $\probp\in\cP$ and thus $\langle X,\mu_{\mathbb{P},Z}\rangle\neq0$ where $Z=\one_{\{X>0\}}$ if $\probp(X>0)>0$ and $Z=\one_{\{X<0\}}$ otherwise.
\end{proof}

\smallskip

In view of the preceding lemma, the announced dual characterization of the Fatou property is a direct consequence of Theorems \ref{s-i functionals} and \ref{theo: dual representation S additive}.

\begin{theorem}
\label{theo: robust fatou}
Let $1\leq p\leq\infty$. If $\rho:L^p_c\to\R\cup\{\infty\}$ is {\bf either} quasiconvex, surplus invariant, and monotone {\bf or}
convex, surplus invariant subject to positivity, monotone, and $S$-additive with $S\in(L^p_c)_+\backslash\{0\}$, then the following
statements are equivalent:
\begin{enumerate}
  \item[(a)] $\rho$ has the Fatou property.
  \item[(b)] $\rho$ has the super Fatou property.
  \item[(c)] $\rho$ is $\sigma(L^p_c,\mathrm{ca}_c^\infty)$ lower semicontinuous.
\end{enumerate}
\end{theorem}

\smallskip

\begin{remark}
(i) The topology $\sigma(L^p_c,\mathrm{ca}_c^\infty)$ can be viewed as a ``robust'' counterpart of the topology $\sigma(L^p,L^\infty)$
in the dominated case.

\smallskip

(ii) The preceding result highlights the power of Theorems \ref{s-i functionals} and \ref{theo: dual representation S additive}. In this case, we do not know how to represent the order-continuous dual $(L^p_c)^\sim_n$ in a tractable way. However, this is not a problem because we can choose an arbitrary separating ideal such as $\mathrm{ca}_c^\infty$ to automatically ensure lower semicontinuity.

\smallskip

(iii) Note that, differently from the dominated case, the (super) Fatou property cannot be expressed in terms of sequences because the
underlying vector lattice fails to satisfy the countable sup property in general. In order to ensure it, one has to impose extra
conditions such as the $\cP$-sensitivity of the sublevel sets of $\rho$, see Maggis et al.~(2018).
\end{remark}

\smallskip

%%%%%%%%%%%%%%%%%%%%%%%%%%%%%%%%%%%%%%%%%%%%%%%%

\subsection{Extensions of surplus-invariant functionals}

The standard theory of risk measures was originally developed in the context of the space $L^\infty$. Since the standard distributions used in finance and insurance, such as normal distributions, are not supported by bounded random variables, it is however critical both from a theoretical and a practical perspective to extend the domain of definition of a risk measure beyond the setting of bounded positions. More specifically, one is interested in looking for the ``canonical'' or ``natural'' model space for a certain class of risk measures, i.e.\ for the largest model space where the defining properties of the given class are still satisfied. This problem has been studied in a number of papers including Delbaen~(2002), Delbaen (2009), Filipovi\'{c} and Svindland (2012), Pichler (2013), Liebrich and Svindland (2017).

\smallskip

The general extension results in Theorems \ref{theo: general extension} and~\ref{theo: extension si subject pos} can be used to show that every surplus-invariant functional defined on $L^\infty$ can be always extended to the entire space $L^0$ without losing monotonicity, (quasi)convexity, and the Fatou property. The same conclusion is true if we weaken surplus invariance to surplus invariance subject to positivity but require $S$-additivity. This unveils a remarkable aspect of surplus-invariant functionals as compared to other classes of risk measures, for which it is typically not possible to ensure extensions beyond the space $L^1$. In the spirit of Filipovi\'{c} and Svindland (2012), one could therefore say that the canonical model space for surplus-invariant functionals is $L^0$.

\smallskip

By virtue of the generality of our lattice approach, we can in fact establish extension results directly in the robust setting with no extra effort. The extension from $L^\infty$ to $L^0$ corresponds to the special situation where $\cP=\{\probp\}$.

\begin{theorem}
\label{theo: robust extensions}
Consider a quasiconvex monotone functional with the Fatou property $\rho:L^\infty_c\to\R\cup\{\infty\}$.

\smallskip

(i) If $\rho$ is surplus invariant, then it can be uniquely extended to a quasiconvex surplus-invariant monotone functional with the
Fatou property $\overline{\rho}:L^0_c\to\R\cup\{\infty\}$. In particular, for every $X\in L^0_c$ we have
\[
\overline{\rho}(X) = \lim_{n\to\infty}\rho(-X^-\wedge n)).
\]
If $\rho$ is additionally convex, then so is $\overline{\rho}$ and for every $X\in L^1_c$ we have
\[
\overline{\rho}(X) = \sup_{\mu\in(\mathrm{ca}_c^\infty)_+}\big(\E_\mu[X^-]-\rho^\ast(\mu)\big),
\]
where
\[
\rho^\ast(\mu) = \sup_{X\in L^\infty_c}\big(\E_\mu[-X]-\rho(X)\big).
\]

\smallskip

(ii) If $\rho$ is $S$-additive with $S\in (L^\infty_c)_+\backslash\{0\}$ and surplus invariant subject to positivity, then it can be
uniquely extended to a convex monotone functional with the Fatou property $\overline{\rho}:L^0_c\to\R\cup\{\infty\}$ that is
$S$-additive and surplus invariant subject to positivity. In particular, for every $X\in L^0_c$ we have
\[
\overline{\rho}(X) = \inf\{m\in\R \,; \ -(X+mS)^-\in\cl_o(\cA_-)\},
\]
where $\cA=\{\rho\leq 0\}$. If $\rho$ is additionally convex, then so is $\overline{\rho}$ and for every $X\in L^1_c$
\[
\overline{\rho}(X) = \sup_{\mu\in(\mathrm{ca}_c^\infty)_+,\,\E_\mu[S]=1}\big(\E_\mu[-X]-\rho^\ast(\mu)\big),
\]
where
\[
\rho^\ast(\mu) = \sup_{X\in L^\infty_c}\big(\E_\mu[-X]-\rho(X)\big).
\]
\end{theorem}
\begin{proof}
Note that, for every $X\in L^0_c$, the supremum of the sequence $(\abs{X}\wedge n\one_\Omega)$ is given by the equivalence class represented by the corresponding pointwise supremum. As a result, $\one_\Omega$ is a weak order unit in $L^0_c$ and the principal ideal generated by $\one_\Omega$ is easily seen to coincide with $L^\infty_c$. The claim follows by applying Theorems \ref{theo: general extension} and~\ref{theo: extension si subject pos}.
\end{proof}

\smallskip

%%%%%%%%%%%%%%%%%%%%%%%%%%%%%%%%%%%%%%%%%%

\subsection{Decomposition of surplus-invariant sets}

In this section we specify the decomposition result established in Theorem \ref{decom} to the framework of random variables. We focus on the general formulation in the context of robust model spaces.

\smallskip

As a preliminary observation, recall that $L^\infty_c$ is order complete by our assumption $L^\infty_c=(\mathrm{ca}_c)^\ast$ and
therefore every ideal of $L^0_c$ is also order complete by Lemma~\ref{lem: characterization order completeness robust}. In particular, every ideal of $L^0_c$ enjoys the projection property. In what follows, for every set $\cA\subset L^0_c$ and every $E\in\cF$ we use the notation
\[
\one_E\cA := \{\one_EX \,; \ X\in\cA\}.
\]

\begin{lemma}
Let $\cX$ be an ideal of $L_c^0$. Then every band of $\cX$ is of the form $\one_E\cX$ for some $E\in\cF$.
\end{lemma}
\begin{proof}
For every $E\in\cF$, since $\cX=\one_E\cX\oplus\one_{E^c}\cX$, it follows that $\cB$ is a (projection) band.
Now, let $\cB$ be a band of $\cX$. Since $L^\infty_c$ is order complete, we have
\[
L^\infty_c = (\cB\cap L^\infty_c)^{dd}\oplus(\cB\cap L^{\infty}_c)^{d},
\]
where the disjoint complements are all taken in $L^\infty_c$. Let $\one_\Omega=U+V$, where $U\in(\cB\cap L^\infty_c)^{dd}_+$ and
$V\in(\cB\cap L^{\infty}_c)^{d}_+$. Setting $E=\{U>0\}$ and using that $U\wedge V=0$, one sees that $U=\one_E$ and $V=\one_{E^c}$. If
$X\in\cB_+$, then $X\wedge\one_\Omega\in\cB\cap L_c^\infty\subset(\cB\cap L^\infty_c)^{dd}$, so that
$0=(X\wedge\one_\Omega)\wedge\one_{E^c}=X\wedge \one_{E^c}$ and therefore $X=X\one_E\in\one_E\cX$. This proves that
$\cB\subset\one_E\cX$. Conversely, take $X\in\one_E\cX$ and assume without loss of generality $X\geq 0$. Since $\cX$ is order complete, we have $\cX=\cB\oplus\cB^d$, where $\cB^d$ is the disjoint complement of $\cB$ in $\cX$. Write $X=Y+Z$, where $Y\in \cB_+$ and $Z\in\cB^d_+$. Note that $Z\wedge\one_\Omega\in L^\infty_c$ and $(Z\wedge\one_\Omega)\wedge\abs{W}=0$ for every $W\in\cB\cap L^\infty_c$. Then, it follows that $Z\wedge\one_\Omega\in(\cB\cap L^{\infty}_c)^{d}$, so that $0=(Z\wedge\one_\Omega)\wedge\one_E=Z\wedge\one_E$. This yields $Z\one_E=0$. Moreover, $0\leq Z\leq X\in \one_E\cX$ implies $Z\one_{E^c}=0$. We thus obtain $Z=0$ and $X=Y\in\cB$. This proves that
$\one_E\cX\subset\cB$, concluding the proof.
\end{proof}

\smallskip

The next decomposition result extends Theorem 3 in Koch-Medina et al.~(2017), which was established for a special class of subspaces of $L^0$ by means of a convenient exhaustion argument. Here, we extend the decomposition to an arbitrary ideal of the robust model space $L^0_c$. In particular, due to the lack of the countable sup property, no exhaustion argument could be applied in our case. The case of an arbitrary ideal of $L^0$ can be simply obtained by setting $\cP=\{\probp\}$.

\begin{theorem}
\label{theo: robust decomposition}
Let $\cX$ be an ideal of $L^0_c$. For every convex order-closed surplus-invariant monotone set $\cA\subset\cX$ we find a measurable
partition $\{E_1,E_2,E_3\}$ of $\Omega$ with
\[
\cA = \one_{E_1}\cX_+\oplus\one_{E_2}(\cX_+-\cD)\oplus\one_{E_3}\cX,
\]
where $\cD$ is a convex order-closed radially-bounded subset of $\one_{E_2}\cX_+$ that is solid in
$\cX_+$ and such that for every $E\in\cF$ with $E\subset E_2$ and $c(E)>0$ there exists $X\in\cD$ for which $\one_EX\neq0$. If $\cA$
is additionally assumed to be a cone, then $c(E_2)=0$.

\smallskip

The event $E_2$ is uniquely determined (up to modifications on $c$-null sets). If $\cX$ contains $\one_\Omega$, then $E_1$ and $E_3$
are also uniquely determined (up to modifications on $c$-null sets).
\end{theorem}
\begin{proof}
Since $L^0_c=(\cX)^{dd}\oplus(\cX)^d$, there exists $E\in\cF$ such that $(\cX)^{dd}=\one_E L^0_c$ and $(\cX)^d=\one_{E^c}L^0$. It
follows that $\cX\subset\one_EL^0_c$ and that, for every $F\in\cF$ with $F\subset E$ and $c(F)>0$, we have $\one_F\wedge X>0$ for some
$X\in\cX$. Now, applying Theorem \ref{decom}, we obtain a unique measurable partition $\{E_1',E_2,E'_3\}$ of $E$ such that
\[
\cA = \one_{E'_1}\cX_+\oplus\one_{E_2}(\cX_+-\cD)\oplus\one_{E'_3}\cX
\]
for some $\cD\subset\one_{E_2}\cX_+$ that satisfies the desired properties. Finally, set $E_1=E_1'\cup E^c$ and $E_3=E'_3$ or
$E_1=E_1'$ and $E_3=E'_3\cup E^c$.
\end{proof}

\smallskip

%%%%%%%%%%%%%%%%%%%%%%%%%%%%%%%%%%%%%%%%%%%%

\subsection{Bipolar theorem for solid subsets of the positive cone}

We conclude this section by showing that our general results on solid sets can be employed to derive a simple ``lattice-based'' proof
of the bipolar theorem by Brannath and Schachermayer (1999). First, we highlight the equivalence between the following closedness
notions for solid sets of positive random variables, which are all understood in a sequential sense. This is a direct consequence of
Lemma \ref{lem: amemiya-fremlin} and Theorem \ref{closed-top}.

\begin{proposition}
\label{prop: closedness solid random variables}
Let $\cX$ be an ideal of $L^0$. For a convex set $\cC\subset\cX_+$ that is solid in $\cX_+$ the following are equivalent:
\begin{enumerate}
  \item[(a)] $\cC$ is closed in $\cX$ with respect to dominated almost-sure convergence.
  \item[(b)] $\cC$ is closed in $\cX$ with respect to almost-sure convergence.
  \item[(c)] $\cC$ is closed in $\cX$ with respect to convergence in probability.
\end{enumerate}
If $\cX^\sim_n$ is separating, the above are also equivalent to:
\begin{enumerate}
  \item[(d)] $\cC$ is $\sigma(\cX,\cI)$ closed for every separating ideal $\cI\subset\cX^\sim_n$.
  \item[(e)] $\cC$ is $\sigma(\cX,\cI)$ closed for some separating ideal $\cI\subset\cX^\sim_n$.
\end{enumerate}
\end{proposition}

\smallskip

In the next statement we adopt the following notation. For a set $\cC\subset L^0_+$ we define
\[
\cC^\circ = \bigcap_{X\in\cC}\{Z\in L^\infty_+ \,; \ \E[ZX]\leq1\} \ \ \ \mbox{and} \
\ \ \cC^\circ_{+} = \cC^\circ\cap\{Z\in L^\infty \,; \ \probp(Z>0)=1\}.
\]

\begin{theorem}[Brannath and Schachermayer (1999)]
\label{theo: bipolar}
Assume that $\cC\subset L^0_+$ is convex, closed with respect to almost-sure convergence and solid in $L^0_+$. Then, we have
\[
\cC = \bigcap_{Z\in\cC^\circ}\{X\in L^0_+ \,; \ \E[ZX]\leq1\}.
\]
If, in addition, $\cC$ is radially bounded, then $\cC^\circ_{+}\neq\emptyset$ and
\[
\cC = \bigcap_{Z\in\cC^\circ_{+}}\{X\in L^0_+ \,; \ \E[ZX]\leq1\}.
\]
\end{theorem}
\begin{proof}
The inclusion ``$\subset$'' in the first identity is clear. To establish the converse inclusion, take
$X\in L^0_+$ and assume that $\E[ZX]\leq1$ for all $Z\in\cC^\circ$.
In particular,
\begin{equation}
\label{eq: bipolar}
\E[Z(X\wedge n)] \leq \E[ZX] \leq 1
\end{equation}
for all $Z\in\cC^\circ$ and $n\in\N$. Note that $\cC\cap L^\infty$
is convex and order closed in $L^\infty$. Moreover, it is also solid in $L^\infty_+$.
Since $L^\infty$ is a separating ideal of $L^1=(L^\infty)^\sim_n$, it follows
from Proposition~\ref{prop: closedness solid random variables} that $\cC\cap
L^\infty$ is $\sigma(L^\infty,L^\infty)$ closed. By applying the standard
bipolar theorem in $L^\infty$ with respect to the locally-convex topology
$\sigma(L^\infty,L^\infty)$, we infer from~\eqref{eq: bipolar} that $X\wedge
n\in\cC\cap L^\infty$ for all $n\in\N$ and thus $X\in\cC$ by order closedness.
This establishes the inclusion ``$\supset$'' and completes the proof of the
first equality.

\smallskip

Now, assume that $\cC$ is radially bounded. Since $\cC^\circ$ is convex, it is easy to see that the second identity will immediately
follow once we show that $\cC^\circ$ contains a strictly-positive random variable. To prove this, note that $\cC\cap L^1$ is convex,
order closed in $L^1$,
radially bounded and solid in $L^1_+$. Hence, we are in a position to apply Lemma~\ref{lem: strictly positive supporting functional}
to ensure the existence of a strictly positive $Z\in L^\infty$ (recall that $(L^1)^\sim_n=L^\infty$) satisfying
\[
\sup\{\E[ZX] \,; \ X\in\cC\cap L^1\} < \infty.
\]
Clearly, up to a simple normalization, we can always assume that
\[
\sup\{\E[ZX] \,; \ X\in\cC\cap L^1\} \leq 1.
\]
Take now an arbitrary $X\in\cC$. Since $\E[Z(X\wedge n)]\leq 1$ for all
$n\in\N$, we infer from the Fatou Lemma that $\E[ZX]\leq 1$ holds as well. This
shows that $Z\in\cC^\circ$ and concludes the proof.
\end{proof}

\smallskip

\begin{remark}
The original bipolar theorem in Brannath and Schachermayer (1999) is stated for convex subsets of $L^0_+$ that are solid in $L^0_+$
and that are closed and bounded in probability. It follows from Theorem 3.4 in Koch-Medina et al.~(2018) that, under convexity and
solidity, boundedness in probability is equivalent to radial boundedness.
\end{remark}

\smallskip

We conclude by showing that our general results allow to deduce the following ``robust'' version of the bipolar theorem. Here, for a set $\cC\subset (L^0_c)_+$ we define
\[
\cC^\circ = \bigcap_{X\in\cC}\{\mu\in(\mathrm{ca}_c^\infty)_+ \,; \ \E_\mu[X]\leq1\}.
\]

\begin{theorem}
\label{theo: robust bipolar}
Assume that $\cC\subset (L^0_c)_+$ is convex, order closed and solid in $(L^0_c)_+$. Then, we have
\[
\cC = \bigcap_{\mu\in\cC^\circ}\{X\in (L^0_c)_+ \,; \ \E_\mu[X]\leq1\}.
\]
\end{theorem}
\begin{proof}
The inclusion ``$\subset$'' is clear. To establish the converse inclusion, take $X\in (L^0_c)_+$ and assume that $\E_\mu[X]\leq1$ for
all $\mu\in\cC^\circ$. In particular,
\begin{equation}
\label{eq: robust bipolar}
\E_\mu[X\wedge n] \leq \E_\mu[X] \leq 1
\end{equation}
for all $\mu\in\cC^\circ$ and $n\in\N$. Note that $\cC\cap L^\infty_c$ is convex and order closed in $L^\infty_c$. Moreover, it is
also solid in $(L^\infty_c)_+$. Since $\mathrm{ca}_c^\infty$ is a separating ideal of $(L^\infty_c)^\sim_n$, it follows from Theorem
\ref{closed-top} that $\cC\cap L^\infty_c$ is $\sigma(L^\infty_c,\mathrm{ca}_c^\infty)$ closed. By applying the standard bipolar
theorem in $L^\infty_c$ with respect to the locally-convex topology $\sigma(L^\infty_c,\mathrm{ca}_c^\infty)$, we infer
from~\eqref{eq: robust bipolar} that $X\wedge n\in\cC\cap L^\infty_c$ for all $n\in\N$ and thus $X\in\cC$ by order closedness. This
establishes the inclusion ``$\supset$'' and completes the proof.
\end{proof}

\smallskip

%%%%%%%%%%%%%%%%%%%%%%%%%%%%%%%%%%%%%%%%%%%%%%

\appendix

\section{Vector Lattices}

In this appendix we recall a variety of notions from the theory of vector
lattices that are freely used in the paper. We refer to the monograph by
Aliprantis and Burkinshaw (2003) for a comprehensive treatment of vector
lattices with a view towards application in economics (and finance).

\smallskip

A real vector space $\cX$ is said to be a {\em vector lattice}, also called a
{\em Riesz space}, if it is equipped with a lattice order $\geq$ such that the
set of {\em positive elements}
\[
\cX_+ = \{X\in\cX \,; \ X\geq0\}
\]
is a convex cone satisfying $\cX_+\cap(-\cX_+)=\{0\}$. For any element $X\in\cX$ we use the standard notation
\[
X^+=X\vee0, \ \ \ X^-=-(X\wedge0), \ \ \ \abs{X}=X\vee(-X).
\]
Note that we have $X=X^+-X^-$ and $\abs{X}=X^++X^-$.
As in Aliprantis and Burkinshaw (2003) or any other standard reference on lattice theory, we make the minor assumption that all the
vector lattices considered are {\em Archimedean}, so that $nX\leq Y$ for all $n\in\N$ implies $X=0$ for every $X,Y\in\cX_+$.

\smallskip

We say that a vector lattice $\cX$ is {\em
order complete} if every subset of $\cX$ that admits an upper bound also admits
a supremum. Every vector lattice $\cX$ admits an {\em order completion} $\cX^\delta$, which is an order-complete vector lattice itself
and contains $\cX$ as an {\em order-dense} sublattice, i.e.~for every $Z\in \cX^\delta_+\backslash\{0\}$, there exists $X\in\cX$ such
that $0<X<Z$. In fact, for every $Z\in\cX^\delta_+\backslash\{0\}$, the supremum of the set $\{X\in \cX \,; \ 0<X<Z\}$ in $\cX^\delta$
is $Z$.
We say that $\cX$ has the {\em countable sup property} if every
subset of $\cX$ that has a supremum admits a countable subset with the same
supremum.

\smallskip

For any subset $\cA\subset\cX$ we use the standard notation
\[
\cA_+=\cA\cap\cX_+ \ \ \ \mbox{and} \ \ \ \cA_-=\cA\cap(-\cX_+).
\]
A subset $\cA\subset\cX$ is said to be {\em solid} whenever
\[
X\in\cA, \ \abs{Y}\leq\abs{X} \ \implies \ Y\in\cA.
\]
If $\cA\subset\cX_+$, we say that $\cA$ is {\em solid in $\cX_+$} whenever
\[
X\in\cA, \ 0\leq Y\leq X \ \implies \ Y\in\cA.
\]
A linear subspace of $\cX$ that is solid is called an {\em ideal}. Note that
every ideal is closed with respect to the lattice operations. For any $X\in\cX$,
the smallest ideal containing $X$ is called the {\em principal ideal of $X$} and
is denoted by
\[
\cI_X = \{Y\in\cX \,; \ \exists \lambda>0 \,:\, \abs{Y}\leq\lambda\abs{X}\}.
\]
Note that every ideal of an order-complete vector lattice is itself order
complete. The countable sup property, if satisfied, also passes to ideals.

\smallskip

For a net $(X_\alpha)\subset\cX$ and an element $X\in\cX$ we write $X_\alpha\uparrow X$ whenever $(X_\alpha)$ is increasing and $\sup_\alpha X_\alpha=X$ in $\cX$. Similarly, we write $X_\alpha\downarrow X$ whenever $(X_\alpha)$ is decreasing and $\inf_\alpha X_\alpha=X$ in $\cX$. We say that a net $(X_\alpha)\subset\cX$ {\em converges in order} to an element $X\in\cX$, written $X_\alpha\xrightarrow{o}X$, if there exists a net
$(Y_\beta)\subset\cX$, possibly over a different index set, such that $Y_\beta\downarrow0$ and for every $\beta$ there exists $\alpha_0$ with
\[
Y_\beta\geq\abs{X_\alpha-X}
\]
for every $\alpha\geq\alpha_0$. Our definition of order convergence is slightly different from that in
Aliprantis and Burkinshaw (2003), but is more common in the current lattice theory. Note that, if $\cX$ is order complete and there
exists $Z>0$ such that $-Z\leq X_\alpha\leq Z$ for every $\alpha$, then $X_\alpha\xrightarrow{o}X$ is equivalent to
\[
X = \sup_\alpha\inf_{\beta\geq\alpha}X_\beta = \inf_\alpha\sup_{\beta\geq\alpha}X_\beta.
\]
Moreover, we say that $(X_\alpha)$ {\em converges in
unbounded order} to $X$, written $X_\alpha\xrightarrow{uo}X$, whenever
\[
\abs{X_\alpha-X}\wedge Y\xrightarrow{o}0
\]
for every $Y\in\cX_+$. Clearly, order convergence implies unbounded-order
convergence but the two convergence notions are not equivalent. The property of
unbounded-order convergence has been thoroughly studied in Gao, Troitsky and
Xanthos (2017). If $\cX$ is an ideal of $L^0(\Omega,\cF,\probp)$, then order convergence of sequences corresponds to dominated
almost-sure convergence and unbounded-order convergence to almost-sure convergence.

\smallskip

A subset $\cA\subset\cX$ is said to be {\em order closed} whenever
\[
(X_\alpha)\subset\cA, \ X_\alpha\xrightarrow{o}X \ \implies \ X\in\cA.
\]
We say that $\cA$ is {\em unbounded-order (uo) closed} if the above property holds by
replacing order convergence with unbounded-order convergence.

\smallskip

An ideal of $\cX$ that is order closed is called a {\em band}. For any
$X\in\cX$, the smallest band containing $X$ is the so-called {\em principal band
of $X$} denoted by
\[
\cB_X = \{Y\in\cX \,; \ \abs{Y}\wedge n\abs{X}\xrightarrow{o}\abs{Y}\}.
\]
The smallest band containing $\cA$ will be denoted by $\band(\cA)$. The {\em
disjoint complement} of $\cA$ is the set
\[
\cA^d = \{X\in\cX \,; \ \abs{X}\wedge\abs{Y}=0, \ \forall Y\in\cA\},
\]
which is easily seen to be a band. A band $\cB$ is said to be a {\em projection
band} provided that $\cX=\cB\oplus\cB^d$. In this case, every $X\in\cX$ can be
uniquely written as $X=X_{\cB}+X_{\cB^d}$ for some $X_{\cB}\in\cB$ and
$X_{\cB^d}\in\cB^d$ and the map $P:\cX\to\cB$ defined by
\[
P(X)=X_{\cB}, \ \ \ X\in\cX,
\]
is called the {\em band projection} of $\cB$. Note that $0\leq P(X)\leq X$ for every $X\in \cX_+$. If every band is a projection band,
we say $\cX$ has the {\em projection property}. Every order-complete vector lattice has the projection property.

\smallskip

A functional $\varphi:\cX\to\R$ is said to be {\em positive} if
$\varphi(X)\geq0$ for every $X\in\cX_+$ and {\em strictly positive} if it is
positive and $\varphi(X)>0$ for every $X\in\cX_+\backslash\{0\}$. We say that
$\varphi$ is {\em order continuous} if it satisfies
\[
X_\alpha\xrightarrow{o}X \ \implies \ \varphi(X_\alpha)\to\varphi(X).
\]
The collection of all linear order-continuous functionals is called the {\em
order-continuous dual of $\cX$} and is denoted by $\cX_n^\sim$. The set
$\cX_n^\sim$ is naturally equipped with a vector lattice structure and is order complete. If $\cX$ is an ideal of
$L^0(\Omega,\cF,\probp)$, the order-continuous dual $\cX_n^\sim$ always has the countable sup property.

\smallskip

A linear topology on $\cX$ is said to be {\em order continuous} if order convergence implies topological convergence and {\em locally solid} if it admits a neighborhood base at zero consisting of solid sets. For any ideal $\cI\subset\cX^\sim_n$ we denote by $\sigma(\cX,\cI)$ the locally-convex topology on $\cX$ generated by the family of seminorms $\{\abs{\varphi(\cdot)} \,; \ \varphi\in\cI\}$. Moreover, we denote by
$\abs{\sigma}(\cX,\cI)$ the locally-convex and locally-solid topology on $\cX$ generated by the family of lattice seminorms $\{\abs{\varphi}(\abs{\cdot}) \,; \ \varphi\in\cI\}$. The topological dual of $\cX$ under any of the above topologies is $\cI$.

%%%%%%%%%%%%%%%%%%%%%%%%%%%%%%%%%%%%%%%%%%%%%%%%%%%%%%%%%

\end{document}